\documentclass[doublecolumn,10pt]{IEEEtran}
\usepackage{setspace}

\usepackage{amsfonts}
\usepackage{amsmath}
\usepackage{amssymb}
\usepackage{mathrsfs}
\usepackage{amstext}
\usepackage{graphicx}
\usepackage{subfigure}

\font\msbm=msbm10

\newtheorem{Theorem}{Theorem}[section]
\newtheorem{lemma}[Theorem]{Lemma}
\newtheorem{corollary}[Theorem]{Corollary}

\newtheorem{proposition}[Theorem]{Proposition}

\newtheorem{remark}[Theorem]{Remark}
\def\mathbb#1{\hbox{\msbm{#1}}}

\newcommand{\ie}{{\em i.e., }}
\newcommand{\tr}{\operatorname{Tr}}
\newcommand{\st}{\operatorname{s.t.} \,}
\newcommand{\bl}{\left(}
\newcommand{\br}{\right)}

\newcommand{\R}{{\mathbb{R}}}

\newcommand{\C}{{\mathbb{C}}}

\renewcommand{\P}{{\mathbb{P}}}
\newcommand{\E}{{\mathbb{E}}}

\newcommand{\sgn}{\operatorname{sgn}}

\newcommand{\beq}{\begin{eqnarray}}
\newcommand{\eeq}{\end{eqnarray}}
\newcommand{\beqn}{\begin{eqnarray*}}
\newcommand{\eeqn}{\end{eqnarray*}}

\newcommand{\supp}{\operatorname{supp}}

\renewcommand{\Re}{\operatorname{Re}}
\renewcommand{\Im}{\operatorname{Im}}

\newcommand{\qed}{\rule{2.5mm}{2.5mm}}

\begin{document}
\title{Average Case Analysis of Multichannel Sparse Recovery Using Convex Relaxation
\thanks{Yonina C. Eldar is with the Technion---Israel
Institute of Technology, Haifa Israel. Email:
yonina@ee.technion.ac.il. Holger Rauhut is with the Hausdorff
Center for Mathematics and Institute for Numerical Simulation,
University of Bonn, Germany. \newline The work of Y.\ Eldar was
supported by the Israel Science Foundation under Grant no. 1081/07
and by the European Commission in the framework of the FP7 Network
of Excellence in Wireless COMmunications NEWCOM++ (contract no.
216715). H.\ Rauhut acknowledges support by the WWTF project
SPORTS (MA 07-004) and by the Hausdorff Center for Mathematics.} }
\author{Yonina C.~Eldar,~\IEEEmembership{Senior~Member,~IEEE} and Holger Rauhut
}

\date{\today}

\maketitle

\begin{abstract}
In this paper, we consider recovery of jointly sparse multichannel
signals from incomplete measurements. Several approaches have been
developed to recover the unknown sparse vectors from the given
observations, including thresholding, simultaneous orthogonal matching
pursuit (SOMP), and convex relaxation based on a mixed matrix
norm. Typically, worst-case analysis is carried out in order to
analyze conditions under which the algorithms are able to recover
any jointly sparse set of vectors. However, such an approach is
not able to provide insights into why joint sparse recovery is
superior to applying standard sparse reconstruction methods to
each channel individually. Previous work considered an average
case analysis of thresholding and SOMP by imposing a probability
model on the measured signals.
In this paper, our main focus is
on analysis of convex relaxation techniques. In particular, we
focus on the mixed $\ell_{2,1}$ approach to multichannel recovery.
We show that under a very mild condition on the sparsity and on
the dictionary characteristics, measured for example by the
coherence, the probability of recovery failure decays
exponentially in the number of channels. This demonstrates that
most of the time, multichannel sparse recovery is indeed superior
to single channel methods. Our probability bounds are valid and
meaningful even for a small number of signals. Using the tools we
develop to analyze the convex relaxation method, we also tighten
the previous bounds for thresholding and SOMP.
\end{abstract}
\vspace{0.5cm}


\section{Introduction}

Recovery of sparse signals from a small number of measurements is
a fundamental problem in many different signal processing tasks
such as image denoising \cite{carota06-1}, analog-to-digital
conversion \cite{ME07,E08,ME09}, radar, compression, inpainting,
and many more. The recent framework of compressed sensing (CS),
founded in the works of Donoho \cite{do06-2}, Cand{\`e}s, Romberg and Tao
\cite{carota06-1}, studies acquisition methods as well as
efficient computational algorithms that allow reconstruction of a
sparse vector $x$ from linear measurements $y=Ax$, where $A\in
\R^{n \times N}$ is referred to as the measurement matrix. The key
observation is that $y$ can be relatively short, so that $n<N$,
and still contain enough information to recover $x$.

Determining the sparsest vector $x$ consistent with the data
$y=Ax$ is generally an NP-hard problem \cite{avdama97}. To
determine $x$ in practice, a multitude of efficient algorithms
have been proposed, \cite{avdama97,efhajoti04,tr04,CRT06,CT05},
which achieve high recovery rates. 
The basis pursuit (BP), or
$\ell_1$-minimization approach, is the most extensively studied
recovery method \cite{chdosa99,carota06-1,do06-2,ra05-7}. The use
of general purpose or specialized convex optimization techniques
\cite{bogokikslu07,efhajoti04} allows for efficient reconstruction
using this strategy. Although greedy methods, such as simple
thresholding or orthogonal matching pursuit (OMP), are faster in
practice, BP provides significantly better recovery guarantees. In
particular, there exist measurement matrices $A \in \R^{n \times
N}$ that allow for stable recovery of all $k$-sparse vectors as
long as $n \geq C k \log(N/k)$ where $C$ is a constant. Such
uniform recovery is not possible for simple thresholding or OMP
\cite{do06,ra07}. (We note, however, that the recent greedy
algorithms CoSaMP \cite{NT08} and ROMP \cite{NV08} are able to
provide such uniform guarantees.) In practice, the recovery rate
of BP when averaged over all random sparse vectors is typically
better than that predicted by the theory. This is due to the fact
that existing analysis considers the ability of BP to recover all
vectors $x$. On the other hand, in random simulations, the
worst-case instance of $x$ typically does not occur. Therefore,
considering the behavior of various recovery methods over random $x$
often leads to more characteristic behavior.

The BP principle as well as greedy approaches have been extended
to the multichannel setup where the signal consists of several
channels with joint sparsity support
\cite{gisttr06,gisttr06-1,fora06,Cotter,Chen, ME08,EB08,EM082}. In
\cite{babadusawa05} the buzzword distributed compressed sensing
was coined for this setup.  An alternative approach is to first
reduce the problem to a single channel problem that preserves the
sparsity pattern, and recover the signal support set; given the
support, the measurements can be inverted to recover the input
\cite{ME08}. A variety of different recovery results have been
established that provide conditions ensuring that the output of
the proposed efficient algorithms coincides with the true signals.
In \cite{Chen} a recovery result was derived for a mixed
$\ell_{p,1}$ program in which the objective is to minimize the sum
of the $\ell_p$-norms of the rows of the estimated matrix whose
columns are the unknown vectors. Recovery results for the more
general problem of block-sparsity were developed in \cite{EM082}
based on the block restricted isometry property (RIP), and in
\cite{EB08} based on mutual coherence. In practice, multichannel
reconstruction techniques perform much better than recovering each
channel individually. However, the theoretical equivalence results
predict no performance gain. The reason is that these results
apply to all possible input signals, and are therefore worst-case
measures. Clearly, if we input the same signal to each channel,
then no additional information on the joint support is provided
from multiple measurements. Therefore, in this worst-case scenario
there is no advantage for multiple channels.

In order to capture more closely the true underlying behavior of
existing algorithms and observe a performance gain when using
several channels, we consider an average-case analysis. In this
setting, the inputs are considered to be random variables.
The idea is to develop conditions on the measurement matrix $A$
such that the inputs can be recovered with high probability given
a certain input distribution.

Recently, there have been several papers that consider sparse
recovery with random ensembles. In \cite{tr06-2} random
sub-dictionaries of $A$ are considered and analyzed. This allows
to obtain average results for BP with a single input channel.
In \cite{KV07}, average-case performance of single channel
thresholding was studied.
In \cite{grrascva07,grmarascva07} extensions to two
multichannel recovery algorithms were developed: thresholding and simultaneous
OMP (SOMP) \cite{grrascva07,grmarascva07}. Under a mild condition
on the sparsity and on the matrix $A$, the
probability of reconstruction failure decays exponentially with
the number of channels.
In the present paper we contribute to this line
of research by analyzing the average-case performance of
multichannel BP, \ie mixed $\ell_{2,1}$-minimization
\cite{gisttr06-1,fora06,EM082,EB08}. The tools we derive in this
context are then also used to slightly improve previous bounds on average
performance of multichannel thresholding and SOMP.

The theoretical average-case results we develop for multichannel
BP are superior to the average bounds developed on thresholding
and SOMP. For an equally mild or even milder condition on the
sparsity and on the matrix $A$, we obtain faster exponential decay
of the failure probability with respect to the number of channels.
Thus, in this sense, the extension of BP to the multichannel case
is superior to existing greedy algorithms, just as in the single
channel setting. Moreover, our
recovery results are applicable also in the single channel case
whereas previous results \cite{grrascva07} require a large number of channels to
yield meaningful (\ie positive) probability bounds (although our
new bound for thresholding generalizing the one in \cite{KV07} does not suffer from this drawback).
Note, however, that in simulations SOMP often exhibits the best
performance. This may be explained by the fact that
the bounds are not tight (at least for SOMP).

To develop our probability bounds, we rely on a new sufficient
condition that ensures recovery of the exact signal set via
$\ell_{2,1}$-minimization. This condition generalizes a result of
\cite{tr05-1,fu04} to the multichannel setting, and is weaker than
existing multichannel recovery conditions. Our average-case
analysis is then carried out assuming that the elements of the
input signal are drawn at random. We prove that under a certain
restriction on $A$ and the sparsity set $S$, the sufficient
condition we develop is satisfied with high probability. The
restriction we impose is that the $\ell_2$-norm of $A_S^\dagger
a_\ell$ over all $\ell$ not in the set $S$ is bounded, where
$a_\ell$ is the $\ell$th column of $A$, and $A_S^\dagger$ is the
pseudo inverse of the restriction of $A$ to the columns in $S$.
This is an improvement over known worst-case recovery conditions
which require a bound on the $\ell_1$-norm \cite{Chen,EB08},
and are therefore stronger. Loosely speaking, we will show that while worst-case
results limit the sparsity level to order $\sqrt{n}$, average-case
analysis shows that sparsity up to order $n$ may enable recovery
with high probability. In terms of RIP
constants, instead of bounding the RIP constant for sparsity sets
of size $2k$, we will only need to consider sets of size $k+1$.

The remaining of the paper is organized as follows. In
Section~\ref{sec:l1l2} we introduce our problem and briefly
summarize known equivalence results between the $\ell_{2,1}$
approach for multichannel recovery and NP-hard combinatorial
optimization that recovers the true signals. A new recovery
condition is derived in Section~\ref{sec:rec}, which is weaker
than previous results, and will be instrumental in developing our
average-case analysis in Section~\ref{sec:avg}. Since the
probability bounds we develop depend on the $2$-norm of
$A_S^\dagger a_{\ell}$, in Section~\ref{sec:bn} we derive several
upper bounds on this norm.
In Section~\ref{sec:compare} we use the tools developed in the
previous section to derive new bounds on the average performance
of thresholding and SOMP, that are tighter than existing results
and also applicable to a broader set of problems. We then compare
our bounds on multichannel BP to these results. Finally, in
Section~\ref{sec:sim} we present several simulations demonstrating
the behavior of the different methods.

Throughout the paper, we denote by $A_S$ the submatrix of $A$
consisting of the columns indexed by $S \subset \{1,\hdots,N\}$,
while $X^S$ is the submatrix of $X$ consisting of the rows of $X$
indexed by $S$. The $\ell$th column of $A$ is denoted by $a_\ell$
or $A_\ell$. For a matrix $A$, $\|A\|_2$ is the spectral norm of
$A$, \ie the largest singular value, and $A^*$ is its conjugate
transpose. The unit sphere in $\R^L$ is defined by $S^{L-1}= \{x
\in \R^L, \|x\|_2 = 1\}$; the complex counterpart is denoted
$S^{L-1}_\C= \{x \in \C^L, \|x\|_2 = 1\}$.

\section{Multichannel $\ell_1$-Minimization}
\label{sec:l1l2}

\subsection{Problem Formulation}

We consider multichannel signal recovery where our goal is to
recover a jointly-sparse matrix $X \in \C^{N\times L}$ from $n$
linear measurements per channel. Here $N$ denotes the signal
length and $L$ the number of channels, \ie the number of signals.
We assume that $X$ is jointly $k$-sparse, meaning that there are
at most $k$ rows in the matrix $X$ that are not identically zero.
More formally, we define the support of the matrix  $X$ as
\begin{equation}
\label{eq:support} \supp X = \bigcup_{\ell=1}^L \supp X_\ell,
\end{equation}
where the support of the $\ell$th column is
\begin{equation}
\label{eq:supportk} \supp X_\ell = \{j, X_{j \ell} \neq 0\}.
\end{equation}
Our assumption is that $\|X\|_0:= |\supp X| \leq k$. The measurements are given by
\begin{equation}
\label{eq:meas} Y = A X,\quad Y \in \C^{n \times L},
\end{equation}
where $A \in \C^{n\times N}$ is a given measurement matrix. Each
measurement vector $Y_\ell=AX_{\ell}$ corresponds to a measurement
of the corresponding signal $X_\ell$.

The natural approach to determine $X$ given $Y$ is to solve the $\ell_0$-minimization
problem
\begin{equation}
\label{eq:NP} \min_{X} \|X\|_0 \quad \st \quad AX=Y.
\end{equation}
However, (\ref{eq:NP}) is NP hard in general \cite{avdama97}.
Several alternative methods have been proposed, that have
polynomial complexity
\cite{gisttr06,gisttr06-1,fora06,Cotter,Chen,
ME08,EB08,EM082,ME08}. A variety of different equivalence results
between the solution of the $\ell_0$-problem and the
output of the proposed efficient algorithm. In \cite{Chen} an equivalence
result was derived for a mixed $\ell_{p,1}$ program in which
the objective is to minimize the sum of the $\ell_p$-norms of the
rows of the estimated matrix whose columns are the unknown
vectors.
 The condition is based on mutual coherence, and turns
out to be the same as that obtained from a single measurement
problem, so that the joint sparsity pattern does not lead to
improved recovery capabilities as judged by this condition.
Recovery results for the more general problem of block-sparsity
were developed in \cite{EM082} based on the RIP, and in
\cite{EB08} based on mutual coherence. Reducing these results to
the multiple measurement vectors (MMV) setting leads again to
conditions that are the same as in the single measurement case. An
exception is the work in \cite{grrascva07,grmarascva07} which
considers average-case performance of thresholding and SOMP. Under
a mild condition on the sparsity and on the matrix $A$, the
probability of reconstruction failure decays exponentially with
the number of channels $L$.
In Section~\ref{sec:compare}
we slightly improve on these bounds using the tools developed in this
paper.

 In
Section~\ref{sec:avg} we follow a similar approach and treat the
average behavior of the mixed $\ell_{2,1}$-minimization program
\cite{gisttr06-1,fora06,EM082} defined by
\begin{equation}\label{l1l2}
\min \|X\|_{2,1} = \sum_{j=1}^N \| X^j \|_2, \quad \mbox{ subject to } AX = Y,
\end{equation}
which promotes joint sparsity, as argued for instance in
\cite{fora06}. In the single channel case $L=1$ this is the usual
BP principle. Therefore, our results can also be used to deduce
the average-case behavior of the BP method. This is in contrast to
\cite{grrascva07}, in which the recovery results derived are not
applicable to the single channel case. As we discuss in
Section~\ref{sec:compare}, our theoretical results are superior to
the previous average-case analysis of \cite{grrascva07} in the
sense that we use an equally mild or even milder condition on the
sparsity and on the matrix $A$, but at the same time get a faster
exponential decay of the failure probability with respect to the
number of channels $L$.


\subsection{Recovery Results}

Recovery results for the program (\ref{l1l2}) were considered
in \cite{Chen,EM082,EB08}. In particular, the lemma below is
derived in \cite{Chen} and follows also from \cite{EB08} where the
more general case of block sparsity is considered.
\begin{proposition}
\label{prop:cond1}
 Let $S \subset \{1,\hdots,N\}$ and suppose that
\begin{equation}\label{cond1}
\|A_S^\dagger a_\ell\|_1 <  1 \quad \mbox{ for all } \ell \notin
S,
\end{equation}
with $A_S^\dagger=(A_S^*A_S)^{-1}A_S^*$ denoting the
pseudo-inverse of $A_S$. Then \eqref{l1l2} recovers all $X \in
\C^{N \times L}$ with $\supp X = S$ from $Y = AX$.
\end{proposition}
Note, that the condition above does not depend on the number of
channels. In the next section we will derive a condition similar
to (\ref{cond1}) that involves the $2$-norm instead of the
$1$-norm, and is therefore weaker (namely, easier to satisfy).

Assuming the columns of $A$ are normalized, $\|a_\ell\|_2 = 1$, we
can guarantee that \eqref{cond1} holds as long as the coherence
$\mu$ of $A$ is small enough, where \cite{DH01}
\begin{equation}
\mu = \max_{j \neq \ell} |\langle a_j, a_\ell \rangle|.
\end{equation}
The following result follows from \cite{EB08} by noting that the
block coherence in this setting is equal to $\mu/d$.
\begin{proposition}
\label{prop:mu} Assume that
\begin{equation}\label{condmu1}
(2k-1) \mu < 1.
\end{equation}
Then \eqref{l1l2} recovers all $X$ with $\|X\|_0 \leq k$ from $Y =
AX$.
\end{proposition}
Under the same conditions as in Propositions~\ref{prop:cond1} and
\ref{prop:mu}, it is shown in \cite{tr04} that BP will recover a
single $k$-sparse vector. Therefore, if (\ref{cond1}) holds, then
instead of solving (\ref{l1l2}) we can use BP on each of the
columns of $Y$.

The coherence is lower bounded by \cite{hest03}
\begin{equation}
\mu \geq \sqrt{\frac{N-n}{n(N-1)}}.
\end{equation}
The lower bound
 behaves like $1/\sqrt{n}$ for large $N$, which
limits the Proposition~\ref{prop:mu} to maximal sparsities $k =
{\cal{O}}(\sqrt{n})$. To improve on this we can generalize
existing recovery results \cite{carota06-1,ca08} based on RIP to
the multichannel setup. The restricted isometry constant $\delta_k$ of a
matrix $A$ is defined to be the smallest constant $\delta_k$ such that
\begin{equation}\label{def:RIP}
(1-\delta_k)\|x\|_2^2 \leq \|Ax\|_2^2 \leq (1+\delta_k)\|x\|_2^2,
\end{equation}
for all $k$-sparse vectors $x$. The next proposition follows from
\cite{EM082}.
\begin{proposition}\label{thm_RIP} Assume $A \in \C^{n\times N}$ with $\delta_{2k} < \sqrt{2} -1$
Let $X \in \C^{N \times L}$, $Y = AX$, and let $\overline{X}$ be the minimizer of
\eqref{l1l2}. Then
\[
\|X-\overline{X}\|_{\mbox{F}} \leq C k^{-1/2}\|X -
\hat{X}^{(k)}\|_{2,1}
\]
where $C$ is a constant, $\|X\|_{\mbox{F}}=\sqrt{\tr(X^*X)}$ is
the Frobenius norm of $X$ 
and $\hat{X}^{(k)}$ denotes the best $k$-term approximation of
$X$, \ie $\supp \hat{X}^{(k)}$ consists of the indices
corresponding to the $k$ largest row norms $\|X^{\ell}\|_2$. In
particular, recovery is exact if $|\supp X| \leq k$.
\end{proposition}
It is well known that Gaussian and Bernoulli random matrices $A
\in \R^{n \times N}$ satisfy $\delta_{2k} \leq \sqrt{2}-1$ with
high probability as long as \cite{badadewa06,cata06}
\begin{equation}\label{RIPn}
n \geq C k \log(N/k).
\end{equation}
For random partial Fourier matrices the respective condition is $n \geq c k \log^4(N)$
\cite{ra06,ruve06}.
Therefore, Proposition~\ref{thm_RIP} allows for a smaller number
of measurements. However, there is still no dependency on the
number of channels. Indeed, under the same RIP condition BP will
recover a single $k$-sparse vector and therefore, as before, BP
may as well be applied to each of the columns of $Y$ individually.

We conclude this overview by stressing that known equivalence
results do not improve on those for single channel sparse
recovery. In \cite{EM082,EB08} equivalence results are derived for
a mixed $\ell_{2,1}$ program when different measurement matrices
$A_i$ are used on each channel. In this case, even worst-case
analysis shows improvement over $L=1$. However, when all
measurement matrices are equal, the recovery conditions do not
show any advantage with multiple signals.

\section{A Recovery Condition}
\label{sec:rec}


 Before turning to analyze the average-case behavior of
\eqref{l1l2}, we first develop a new condition on $A$ that allows
for perfect recovery. This formulation will be useful in deriving
the average-case results.

In the following theorem we give a sufficient condition on the
minimizers of \eqref{l1l2}. This theorem generalizes a result of
\cite{tr05-1,fu04} for the $L=1$ case. To this end we denote by
$\sgn(X) \in \C^{N \times L}$ the matrix with entries
\begin{equation}
\sgn(X)_{\ell j} = \left\{\begin{array}{cc}
\frac{X_{\ell j}}{\| X^\ell \|_2}, & \| X^\ell \|_2 \neq 0; \\
0, & \| X^\ell \|_2=0.
\end{array}\right.
\end{equation}
In this definition, each element of $X$ is normalized by the norm
of the corresponding row. When $L=1$, $\sgn(X)$ reduces to the
sign of the elements of the vector $x$.
\begin{Theorem}
\label{thm:sgn} Let $X \in \C^{N \times L}$ with $\supp X = S$ and
assume $A_S$ to be non-singular. If there exists a
matrix $H \in \C^{n \times L}$ such that
\begin{equation}
\label{eq:c1}
 A_S^* H = \sgn(X^S),
\end{equation}
and
\begin{equation}
\label{eq:c2} \|H^* a_\ell \|_2 < 1 \quad \mbox{ for all } \ell
\notin S
\end{equation}
then $X$ is the unique solution of \eqref{l1l2}.
\end{Theorem}
Before proving the theorem we note that the two conditions on $H$
easily imply that
\begin{equation}
\label{eq:normh} \|H^*a_{\ell}\|_2  \leq 1,\quad \mbox { for all }
\ell.
\end{equation}

\begin{proof}
The proof follows the ideas of \cite{tr05-1}, with appropriate
modifications to account for the mixed $\ell_{2,1}$ norm that
replaces the $\ell_1$ norm.

Let $Y=AX$, and assume there exists a matrix $H$ such that $X,H$
satisfy (\ref{eq:c1}) and (\ref{eq:c2}). Let $X'$ be an
alternative matrix satisfying $Y=AX'$. Our goal is to show that
$\|X\|_{2,1}<\|X'\|_{2,1}$. To this end, we note that
\begin{equation}
\label{eq:pr1t}
\|X\|_{2,1}=\|X^S\|_{2,1}=\tr\bl\sgn(X^S)(X^S)^*\br,
\end{equation}
where $\tr$ denotes the trace. Substituting
$A_S^* H = \sgn(X^S)$ into (\ref{eq:pr1t}), and using the cyclicity of the trace
we have
\begin{align}
\label{eq:pr1} \|X\|_{2,1} &= \tr\bl H^*A_S X^S\br=\tr\bl H^*A X'\br\\
&= \tr\bl X'^{S'}H^*A_{S'} \br,\notag
\end{align}
where we used the fact that $A_SX^S=Y=AX'$ and $S'$ denotes the
support of $X'$.
We next rely on the following lemma.
\begin{lemma}
\label{lemma:tr} Let $A, B$ be matrices such that $AB$ is defined.
Then $|\tr(BA)| \leq \|B\|_{2,1} \max_{\ell}
\|A_\ell\|_2$, with strict inequality if $\|A_\ell\|_2 < \max_{\ell} \|A_{\ell}\|_2$
for some value of $\ell$ for which $\|B^\ell\|_2 \neq 0$.
\end{lemma}

\begin{proof}
The proof follows from noting that
\begin{align}
&|\tr(BA)|
\leq  \sum_{\ell} |B^\ell
A_\ell| \leq \sum_{\ell} \|B^\ell\|_2 \|A_\ell\|_2 \notag\\
 &\hspace*{0.2in} \leq \max_{\ell}\|A_\ell\|_2 \sum_{\ell}
\|B^\ell\|_2 = \max_{\ell}\|A_\ell\|_2 \|B^\ell\|_{2,1},\notag
\end{align}
where the second inequality is a result of applying
Cauchy-Schwartz. Under the condition of the lemma, we have strict
inequality in the last inequality.
\end{proof}
Applying Lemma~\ref{lemma:tr} to (\ref{eq:pr1}), leads to
\begin{align}
\label{eq:pr2} \|X\|_{2,1} &\leq \|X'^{S'}\|_{2,1}\max_{\ell \in
S'}\|H^*A_{\ell}\|_2 \leq \|X'^{S'}\|_{2,1} \\
&=\|X'\|_{2,1},\notag
\end{align}
where the last inequality follows from (\ref{eq:normh}). We have
strict inequality in the first inequality of (\ref{eq:pr2}) as
long as the values $\|H^*A_{\ell}\|_2$ for $\ell \in S'$ are not
all equal since $\|{X'}^\ell\|_2 \neq 0$ for all $\ell \in S'$ be definition of the support.

Suppose to the contrary that $\|H^*A_{\ell}\|_2=a$ for all $\ell
\in S'$. Clearly, $S'$ must contain at least one index $\ell$ that
is not contained in $S$; otherwise $S'\subset S$, which would contradict the hypothesis
that $A_S$ is non-singular, $A_{S'} X' = A_S X$ and $X \neq X'$.
By our assumption $\|H^*a_{\ell}\|_2<1$,
which then implies that $a<1$ or $\|H^*A_{\ell}\|_2<1, \ell \in
S'$. The inequalities in (\ref{eq:pr2}) then become
\begin{equation}
\label{eq:pr3} \|X\|_{2,1} \leq \|X'^{S'}\|_{2,1}\max_{\ell \in
S'}\|H^*A_{\ell}\|_2 <\|X'^{S'}\|_{2,1}=\|X'\|_{2,1}.
\end{equation}
Thus, we have shown that $\|X'\|_{2,1}>\|X\|_{2,1}$ for any $X'$
such that $Y=AX'$, and therefore (\ref{l1l2}) recovers the true
sparse matrix $X$.
\end{proof}

Choosing $H = (A_S^\dagger)^* \sgn(X_S)$ in Theorem~\ref{thm:sgn}
results in the following corollary.
 \begin{corollary}\label{cor_reconstruct} Let $X \in \C^{N \times L}$ with $\supp X = S$ and assume $A_S$ to be non-singular.
 If
 \begin{equation}
 \label{eq:conds}
 \|\sgn(X^S)^* A_S^\dagger a_\ell \|_2 < 1 \quad \mbox{ for all } \ell \notin
 S,
\end{equation}
then $X$ is the unique minimizer of \eqref{l1l2}.
 \end{corollary}
 This corollary will be instrumental in proving the average-case
 performance of (\ref{l1l2}). It can easily be seen that
 Corollary~\ref{cor_reconstruct} implies
 Proposition~\ref{prop:cond1}. This follows from the triangle inequality,
\begin{align}
&\left\| \sgn(X^S)^* A_S^\dagger a_\ell \right\|_2
=  \left\|
\sum_{j \in S} (A_S^\dagger a_\ell)_j \sgn(X^j)^* \right\|_2 \notag\\
&\hspace*{0.2in} \leq \sum_{j \in S} |(A_S^\dagger a_\ell)_j|~\|
\sgn(X^j) \|_2 = \|A_S^\dagger a_\ell\|_1,\notag
\end{align}
where we used the fact that $\| \sgn(X^j) \|_2=1$.

\section{Average Case Analysis}
\label{sec:avg}

Intuitively, we would expect multichannel sparse recovery to
perform better than single channel recovery. However, in the worst
case setting this is not true as already suggested by the results
of Section~\ref{sec:l1l2}. The reason is very simple. If each
channel carries the same signal, $X_{\ell} = x$ for
$\ell=1,\hdots,L$, then also the components of $Y = AX$ are all
the same and we do not have more information on the support of $X$
than provided by a single component $Y_{\ell}$. The following
proposition establishes formally that if BP fails for a given
measurement matrix $A$, then multichannel optimization
(\ref{l1l2}) will fail as well so that in the worst-case, adding
channels will not improve performance.
\begin{proposition}\label{prop:worst}
 Suppose there exists a $k$-sparse vector $x \in \C^N$ that $\ell_1$-minimization
is not able to recover from $y = Ax$. Then $\ell_{2,1}$-minimization fails to recover
$X = (x|x|\cdots|x) \in \C^{N\times L}$ from $Y = AX$.
\end{proposition}





\begin{proof}
If $\ell_1$-recovery fails on some
$k$-sparse $x$ then necessarily $\|x'\|_{1} \leq \|x\|_{1}$ for
some $x'$ satisfying $A x' = Ax$.
Clearly $X=(x|x|\cdots|x)$ is (jointly) $k$-sparse and $AX = A X'$ for
$X'=(x'|xÔ|\cdots|x')$. Furthermore,
\[
\|X'\|_{2,1} = \sqrt{L} \|x'\|_1 \leq \sqrt{L} \|x\|_1 = \|X\|_{2,1}
\]
and therefore $X$ is not the unique minimizer of the
$\ell_{2,1}$-minimization problem.
\end{proof}

Realizing that \eqref{l1l2} is not more powerful than usual BP in
the worst case, we seek an average-case analysis. This means that
we impose a probability model on the $k$-sparse $X$. In
particular, as in \cite{grrascva07}, we will assume that on the
support $S$ of size $k$ the coefficients of $X$ are chosen at random.
We then show that under a suitable probability model
on the non-zero elements of $X$,
the condition given by Corollary~\ref{cor_reconstruct} is
satisfied with high probability, which depends on $L$.

We follow the probability model used in \cite{grrascva07}: let $S$
be the joint support of cardinality $k$. On $S$ the coefficients
are given by
\begin{equation}\label{prob:model}
X^S = \Sigma \Phi
\end{equation}
where $\Sigma = \operatorname{diag}(\sigma_j, j\in S) \in
\R^{k\times k}$ is an arbitrary diagonal matrix with positive
diagonal elements $\sigma_j$. The matrix $\Phi$ will be chosen
at random according to one of the following models.
\begin{itemize}
\item {\bf Real Gaussian:} each entry of $\Phi \in \R^{k \times L}$ is chosen independently from a standard normal distribution.
\item {\bf Real spherical}: the rows of $\Phi \in \R^{k \times L}$ are chosen independently and uniformly at random from the real sphere $S^{L-1}$.
\item {\bf Complex Gaussian}: the real and imaginary parts of each entry of
{$\Phi \in \C^{k \times L}$} are chosen independently according to a standard normal distribution.
\item {\bf Complex spherical}: the rows of $\Phi \in \C^{k\times L}$ are chosen
independently and uniformly at random from the complex
sphere $S_\C^{L-1}$. 
\end{itemize}
Note that taking $\Sigma$ to be
the identity matrix results in a standard Gaussian random matrix $X^S$,
while taking arbitrary non-zero $\sigma_j$'s on the diagonal of
$\Sigma$ allows for different variances. The matrix $\Sigma$ may
be deterministic or random. In particular, choosing $\Sigma$ to be
the matrix with diagonal elements given by the inverse
$\ell_2$-norm of the rows of $\Phi$ in the real (complex) Gaussian model,
leads to a matrix $X^S$ with a real (complex) spherical distribution.

In Theorems~\ref{thm_average2} and \ref{thm_average} below we
develop conditions under which (\ref{l1l2}) recovers $X$ from
$Y=AX$ with probability that decays exponentially with $L$. The
condition in both theorems is given in terms of an upper bound on
$\|A_S^\dagger a_\ell\|_2$ for $\ell$ not in $S$. This is in
contrast to the worst-case result of Proposition~\ref{prop:cond1}
that is given in terms of $\|A_S^\dagger a_\ell\|_1$ and therefore
stronger. 
The
essential idea in both proofs is to show that if the bound on
$\|A_S^\dagger a_\ell\|_2$ is satisfied, then the sufficient
condition of Corollary~\ref{cor_reconstruct} holds with high
probability.

Before stating the first theorem, we derive the following result
on the norm of sums of independent random vectors, uniformly
distributed on a sphere.
\begin{Theorem}\label{thm:Bernstein} Let $a \in \C^k$ and let $Z_j$, $j=1,\hdots,k$, be a sequence
of independent random vectors which are uniformly distributed on the real sphere
$S^{L-1}$. 
Then for any $u > 1$
\begin{align}
&\P\bl\left\| \sum_{j=1}^k a_j Z_j\right\|_2 \geq u\|a\|_2\br \notag\\
&\hspace*{0.2in} \leq \exp\left(- \frac{L}{2}(u^2 -
\log(u^2)-1)\right).\notag
\end{align}
\end{Theorem}

\begin{proof}
See Appendix~\ref{sec:prob}.
\end{proof}

 Theorem~\ref{thm:Bernstein}
generalizes the Bernstein inequality for Steinhaus sequences in
\cite[Theorem 13]{tr06-2} to higher dimensions. We may extend the
estimate easily to random vectors
uniformly distributed on complex unit spheres.

\begin{corollary}\label{cor:Bernstein} Let $a \in \C^k$ and let $Z_j$, $j=1,\hdots,k$, be a sequence
of independent random vectors which are uniformly distributed on the complex
sphere $S^{L-1}_\C$. 
Then for any $u > 1$
\begin{align}
&\P\bl\left\| \sum_{j=1}^k a_j Z_j\right\|_2 \geq u\|a\|_2\br \notag\\
&\hspace*{0.2in} \leq \exp\left(- L(u^2 -
\log(u^2)-1)\right).\notag
\end{align}
\end{corollary}

\begin{proof} First observe that $a_j Z_j$ has the same distribution
as $|a_j| Z_j$. We may therefore assume without loss of generality
that $a_j \in \R$. Next, a random vector $Z \in S_\C^{L-1}$ is
uniformly distributed on $S_\C^{L-1}$ if and only if $(\Re (Z)^T,
\Im(Z)^T )^T$ is uniformly distributed on the real sphere
$S^{2L-1}$.
Applying Theorem~\ref{thm:Bernstein} with $L$ replaced by $2L$
yields the statement.
\end{proof}

With this tool at hand we can now easily prove the following
average-case recovery theorem.

\begin{Theorem}\label{thm_average2}
Let $S \subset \{1,\hdots,N\}$ be a set of cardinality $k$
and suppose
\begin{equation}\label{cond:boundnorm}
\|A_S^\dagger a_\ell\|_2 \leq \alpha < 1 \quad \mbox{for all }
\ell \notin S.
\end{equation}
Let $X \in \R^{N\times L}$ with $\supp X \subset \{1,\hdots, N\}$ such
that the coefficients on $S$ are given by \eqref{prob:model} with
some diagonal matrix $\Sigma \in \R^{k \times k}$ and $\Phi \in \R^{k \times L}$
chosen from the real Gaussian or spherical probability.
Then with probability at least
\begin{align}\label{prob:bound}
1- N \exp\left(-\frac{L}{2}(\alpha^{-2} - \log(\alpha^{-2}) - 1)\right)
\end{align}
\eqref{l1l2} recovers $X$ from $Y=AX$.

If the real probability model is replaced by one of the two complex models
then $L/2$ can be replaced by $L$ in \eqref{prob:bound}.
\end{Theorem}
For $\alpha<1$ we are guaranteed that the exponent in \eqref{prob:bound} has a negative
argument, and therefore the error decays exponentially in $L$.

\begin{proof} First observe that by the rotational invariance of Gaussian random vectors
the columns of $\sgn(X^S)^* = \sgn(\Phi^*)$ are independent and
uniformly distributed on the real sphere, and the same is also true
if we use the real spherical random model. Denote $b^{(\ell)} =
A_S^\dagger a_\ell$ for $\ell \notin S$ and by $Z_j$,
$j=1,\hdots,k$ a sequence of independent random vectors that are
uniformly distributed on the sphere $S^{L-1}$. Using the
sufficient recovery condition of Corollary \ref{cor_reconstruct},
the union bound and Theorem~\ref{thm:Bernstein} we can estimate
the probability that $\ell_{2,1}$ minimization fails to recover
$X$ by
\begin{align}
&\hspace*{-0.2in} \P(\max_{\ell \notin S} \|\sgn(X^S)^*b^{(\ell)}\|_2 > 1)\notag\\
&\leq \sum_{\ell \notin S} \P(\|\sgn(X^S)^* b^{(\ell)}\|_2
> 1) \notag\\
& \leq \sum_{\ell \notin S} \P\bl\left\| \sum_{j = 1}^k
b^{(\ell)}_j Z_j\right\|_2 > \alpha^{-1} \|b^{(\ell)}\|_2 \br \notag\\
&\leq
(N-k) \exp\left(-\frac{L}{2}(\alpha^{-2} - \log(\alpha^{-2}) -
1)\right).\notag
\end{align}
The complex case follows analogously using  Corollary \ref{cor:Bernstein}.
\end{proof}
For $L=1$, Theorem~\ref{thm_average2} is contained implicitly in
\cite[Theorem 13]{tr06-2}.
The appearance of the $2$-norm in \eqref{cond2} instead of the
$1$-norm as in \eqref{cond1} makes the condition of the theorem
weaker than worst-case estimates
(recall that $\|x\|_2 \leq \|x\|_1 \leq \sqrt{k} \|x\|_2$ for any length-$k$
vector $x$). In Section~\ref{sec:bn} this will be made more evident
when we consider conditions on the coherence $\mu$ and the RIP
constant to allow for recovery with high probability. The
requirement we obtain on $\mu$ is weaker than that of
Proposition~\ref{prop:mu} and allows for recovery with $k$ on the
order of $n$, while the worst-case results limit recovery to order
$\sqrt{n}$. Furthermore, in contrast to the worst-case results
which depend on $\delta_{2k}$, we will show that high-probability
recovery is possible as long as $\delta_{k+1}$ is small enough.

It is evident from \eqref{prob:bound} that the failure probability
decays exponentially with growing number of channels $L$.
Moreover, the bound is also useful for small $L$, and in
particular for the monochannel case $L=1$. Indeed, a simple
algebraic manipulation shows that the failure probability is less
than $\epsilon$ provided $\|A_S^\dagger a_\ell\|_2 \leq \alpha$
for all $\ell \notin S $ with $\alpha$ satisfying
\[
\alpha^{-2} - \log(\alpha^{-2}) \geq \frac{2 \log(N/\epsilon)}{L} + 1.
\]
This provides a useful average-case analysis even for $L=1$.

For completeness, we also state an alternative recovery result below
which provides a slightly better probability estimate
than Theorem \ref{thm_average2} for very large values of $N$. However,
the required condition on $\|A_S^\dagger a_\ell\|_2$ is stronger.




\begin{Theorem}\label{thm_average}
Let $S \subset \{1,\hdots,N\}$ be a set of cardinality $k$, and
let $X \in \R^{N\times L}$ be random sparse coefficients
with $\supp X = S$ given by the real Gaussian probability model.
If
\begin{equation}\label{cond2}
\|A_S^\dagger a_\ell\|_2 < \frac{A_L}{3\sqrt{L}+2\sqrt{k}} =:\gamma \sim
\frac{1}{3+2\sqrt{k/L}} 
\end{equation}
for all $\ell \notin S$,
where
\begin{equation}\label{def:AL}
A_L=\sqrt{2}\frac{\Gamma((L+1)/2)}{\Gamma(L/2)} \sim \sqrt{L},
\end{equation}
and $\Gamma$ denotes the Gamma function, then with probability at
least
\[
P=1 - \exp(-L/8) - k\exp(-A_L^2/8) 
\]
\eqref{l1l2} recovers $X$ from $Y = AX$.
\end{Theorem}

It follows from Stirling's formula
$\Gamma(z) \sim \sqrt{2\pi z} z^{z-1/2} e^{-z}$,
that
\begin{align}
&\hspace*{-0.1in} A_L = \sqrt{2}
\frac{\Gamma((L+1)/2)}{\Gamma(L/2)} \sim \sqrt{2}
\frac{((L+1)/2)^{L/2} e^{-(L+1)/2}}{(L/2)^{(L-1)/2}e^{-L/2}}
\notag\\
& \hspace*{0.08in} = e^{-1/2}\frac{(L+1)^{L/2}}{L^{(L-1)/2}} =
e^{-1/2} \left(L (1+1/L)^L\right)^{1/2}
 \sim \sqrt{L}.\notag
\end{align}
Moreover, for all
$L\geq 1$ it holds that $\sqrt{L} \geq A_L \geq \sqrt{ \frac{2}{\pi}} \sqrt{L}
\approx 0.797 \sqrt{L}$.


Note that
$\gamma =\frac{A_L}{3\sqrt{L}+2\sqrt{k}}$
 is monotonically increasing in $L$.
 In addition, the probability $P$ is also increasing (towards $1$) in
 $L$. Therefore, more channels increase the
 probability of success and in addition relax the requirements on
 the matrix $A$.

\begin{proof} To prove the theorem we show that if (\ref{cond2})
is satisfied, then condition (\ref{eq:conds}) of
Corollary~\ref{cor_reconstruct} holds with probability $P$.

To this end, let $\Phi \in \R^{k \times L}$ denote a random matrix
with independent standard normal distributed entries, and define
$D$ as the $k \times k$ diagonal matrix with
diagonal elements $1/s_j,j \in S$, where
$
s_j = \|\Phi^j\|_2 = \sqrt{\sum_{\ell=1}^L |\Phi_{j\ell}|^2}.
$
We can then express 
$\sgn(X^S) = \sgn(\Sigma \Phi)
= \sgn(\Phi) = D \Phi$. (This equation also means that the
diagonal matrix $\Sigma$ does not play any role.) Denoting $b_j =
A_{T}^\dagger a_j$ for  $j \notin S$,
\[
\|\sgn(X_S)^* b_j\|_2 = \| \Phi^* D b_j\|_2 \leq \|\Phi\|_{2}
\|D\|_{2} \|b_j\|_2.
\]
By the assumption of the theorem $ \|b_j\|_2<\gamma$ where
$\gamma$ is defined by 
(\ref{cond2}). It therefore remains to
bound $\|\Phi\|_2$ and $\|D\|_2$.
>From \cite[equation (4.35)]{cata06}, see also \cite{sz91}, the
operator norm of $\Phi$ satisfies
\begin{equation}\label{opnorm}
\|\Phi\|_{2} \leq \sqrt{L} + \sqrt{k} + r
\end{equation}
with probability
at least $1 - \exp(- r^2/2)$.

Next we consider $\|D\|_2$. Observe that the $s_j^2$ are $\chi^2(L)$
distributed.
Therefore, denoting a $\chi^2(L)$-variable by $Y$,
\begin{align}
&\E[s_j] = \E[ \sqrt{Y}] = \frac{1}{2^{L/2}\Gamma(L/2)} \int_0^\infty \sqrt{x} x^{L/2}e^{x/2} dx\notag\\
&\hspace*{0.28in} = \sqrt{2}\frac{\Gamma((L+1)/2)}{\Gamma(L/2)} =
A_L \sim \sqrt{L}.\notag
\end{align}
As a function of $\Phi^j$ the $s_j$ are Lipschitz continuous, \ie
$s_j(\Phi^{j} - \Psi^{j}) \leq \|\Phi^{j}-\Psi^{j}\|_2$.
Using these two observations we rely on the following standard
concentration of measure result, see e.g.~\cite[eq.~(2.35)]{le01}
or \cite[eq.~(1.6)]{leta91}.
\begin{Theorem}\label{thm_conc} Let $f$ be a Lipschitz function on
$\R^L$, \ie 
$|f(x)-f(y)| \leq B\|x-y\|_2$
for all $x,y \in \R^{L}$. Further assume
that $Z = (Z_1,Z_2,\hdots,Z_L)$ is a vector of independent standard Gaussian
random variables. Then
\begin{align}
\P(f(Z) \geq  \E[f(Z)] + t) & \leq
\exp\left(-\frac{t^2}{2B^2}\right), \nonumber \\
\quad \P(f(Z) \leq  \E[f(Z)] - t)& \leq
\exp\left(-\frac{t^2}{2B^2}\right).\notag
\end{align}
\end{Theorem}
Our goal is to show that $\|D\|_2$ is bounded from above, which is
equivalent to bounding the smallest value of $s_j$ from below.
Applying Theorem~\ref{thm_conc} to $s_j$,
\[
\P(s_j < A_L(1 - t)) \leq \exp(-t^2A_L^2/2),
\]
where we used the fact that $B=1$ and $\E[s_j]=A_L$.
Using a union bound over all $j$, we obtain
\begin{align}
&\hspace*{-0.2in} \P(s_j < A_L(1-t),\forall j) =
\P\left(\min_{j=1,\hdots,k}
s_j < A_L(1-t)\right)  \nonumber \\
& \leq  \sum_{j \in S} \P(s_j < A_L(1-t)) 
 =  k \exp(-t^2 A_L^2/2). \notag 
\label{minsj}
\end{align}
Assuming that $\min_{j\in S} s_j \geq A_L(1-t)$ holds, $\|D\|_{2}
\leq 1/(A_L(1-t))$. Combining this bound with \eqref{opnorm} for
$r = \sqrt{L} s$ we have
\begin{align}
&\hspace*{-0.2in} \|\sgn(X_S) A_S^\dagger a_j\|_2 \leq
\frac{\sqrt{k} + \sqrt{L} +
s\sqrt{L}}{A_L(1-t)}\gamma \notag\\
&= \frac{(s+1 +
\sqrt{k/L})\gamma\sqrt{L}}{(1-t)A_L}.\notag
\end{align}
Choosing $s=t=1/2$,
\begin{equation}\label{est1}
\|\sgn(X_S)^* A_S^\dagger a_j \|_2 \leq (3+2\sqrt{k/L}) \gamma
\sqrt{L}/A_L < 1.
\end{equation}
>From (\ref{est1}) and Corollary \ref{cor_reconstruct}, $X$ is
recoverable using (\ref{l1l2}).

The probability that (\ref{est1}) does not hold can be computed by
applying a union bound to the probabilities that the spectral
norms of each of the matrices $\Phi$ and $D$ are not bounded. This
shows that (\ref{est1}) does not hold with probability at most
$
\exp(-L/8) + k \exp(-A_L^2/8) 
$
completing the proof of the theorem.
\end{proof}

\section{Bounded Norm Condition}
\label{sec:bn}

Both Theorems~\ref{thm_average2} and \ref{thm_average} state that
$X$ can be recovered with high probability from $Y$, as long as
$\|A_S^\dagger a_\ell\|_2$ is bounded. In this section we develop
several different conditions under which this holds.

\begin{proposition}\label{prop:bound:mu} Let $A \in \C^{n \times N}$
have unit-norm columns and coherence $\mu$, and let $S \subset
\{1,\hdots,N\}$ be a set of cardinality $k$. Assume that
\begin{equation}\label{condmu2}
(\sqrt{k} + (k-1)\delta) \mu < \delta
\end{equation}
for some $\delta > 0$. Then $\|A_S^\dagger a_\ell\|_2 \leq \delta$
for all $\ell \notin S$.
\end{proposition}

\begin{proof}
Gershgorin's disk theorem implies that the smallest eigenvalue
$\lambda_{\min}$ of $A_S^* A_S$ is bounded from below by
$1-(k-1)\mu$. In particular, $A_S^*A_S$ is invertible provided
$(k-1)\mu < 1$. Further, \[ \|A_S^* a_\ell\|_2 = \sqrt{\sum_{j \in
S} |\langle a_\ell, a_j \rangle|^2} \leq \sqrt{k} \mu,\] since by
definition, $|\langle a_\ell, a_j \rangle| \leq \mu$. Now, using
the fact that  $A_S^\dagger = (A_S^* A_S)^{-1} A_S^*$,
\begin{align}
\|A_S^\dagger a_\ell\|_2 &\leq \|(A_S^* A_S)^{-1}\|_{2}
\|A_S^*a_\ell\|_2\notag\\
&\leq (1-(k-1)\mu)^{-1} \sqrt{k} \mu < \delta,\notag
\end{align}
where the last inequality follows from the fact that
(\ref{condmu2}) implies $\delta>\sqrt{k}/(1-(k-1)\mu)^{-1}$.
\end{proof}

Condition \eqref{condmu2} is slightly weaker than \eqref{condmu1}
as long as $\delta>1/\sqrt{k}$. This follows from the $2$-norm
that replaced the $1$-norm in the upper bound. However, \eqref{condmu2} still
suffers the square-root bottleneck $k = {\cal O}(\sqrt{n})$. To
improve on this result, we next provide a condition based on the
following refinement of the RIP of $A$.
For a set $S \subset \{1,\hdots,N\}$ we let
\[
\delta(S) = \|A_S^* A_S - I\|_2.
\]
The restricted isometry constant $\delta_k$ of \eqref{def:RIP}
satisfies $ \delta_k = \max_{|S|\leq k} \|A_S^* A_S - I\|_2 $ so
that if $S$ has cardinality $k$ then $\delta(S) \leq \delta_k$. We
further define
\begin{equation}
\label{eq:deltas} \delta^*(S) = \max_{\ell \notin S} \delta(S \cup
\{\ell\}).
\end{equation}
Clearly, $\delta(S) \leq \delta^*(S) \leq \delta_{k+1}$. Finally,
we make use of the following ``local" $2$-coherence function,
\begin{equation}\label{def:mu2}
\mu_2(S) = \max\left\{ \max_{\ell \notin S} \|A_S^* a_\ell\|_2,
\max_{\ell \in S} \|A_{S \setminus \ell}^* a_\ell\|_2 \right\}
\end{equation}
for a subset $S \subset \{1,\hdots,N\}$, where $S \setminus \ell$
denotes the elements in $S$ excluding the $\ell$th one.
 From the
definition of the coherence it follows immediately that
\begin{equation}\label{mudelta:est1}
\mu_2(S) \leq \sqrt{|S|} \mu,
\end{equation}
since the magnitude of each element $|\langle a_\ell, a_j\rangle|$
of the vector $A_S^*a_\ell$ is bounded above by $\mu$. In
addition,
\begin{equation}\label{mudelta:est}
\mu_2(S) \leq \delta^*(S).
\end{equation}
This is a result of the fact that $A_S^*a_\ell$ is a submatrix of
$A_{S \cup \{\ell\}}^* A_{S \cup \{\ell\}} - I$
 for $\ell \notin
S$, while $A_{S \setminus \{\ell\}}^* a_\ell$ is a submatrix of
$A_S^* A_S - I$ for $\ell \in S$. (They both consist of a
subcolumn of the respective matrix, that ``leaves" out the
diagonal element.) We now use these definitions to bound
$\|A_S^\dagger a_\ell\|_2$:
\begin{proposition}\label{lem_RIP}
Let $S\subset \{1,\hdots,N\}$. Then: 
\begin{itemize}
\item[(a)]
If $A$ satisfies $\delta^*(S) \leq \delta < 1/2$ 
then
\[
\|A_S^\dagger a_\ell\|_2 \leq \frac{\delta}{1-\delta} < 1\quad
\mbox{ for all }\ell \notin S.
\]
\item[(b)] If $A$ satisfies $\delta(S) \leq \delta < 1$ and
$\mu_2(S) \leq \eta$ then
\[
\|A_S^\dagger a_\ell\|_2 \leq \frac{\eta}{1-\delta}.
\]
\end{itemize}
\end{proposition}

\begin{proof} 
Denoting by $\lambda$ an
eigenvalue of $A_S^*A_S$, the definition of $\delta(S) \leq \delta^*(S) \leq \delta$ implies that $|1-\lambda| \leq
\delta$. Consequently, the smallest eigenvalue of $A_S^* A_S$ is
bounded from below by $1-\delta$ and therefore
\[
 \|(A_S^* A_S)^{-1}\|_{2} \leq
\frac{1}{1-\delta}.\]

For (a), as already noted above, $A_S^* a_\ell$ for $\ell \notin S$ is a
$k \times 1$ submatrix of $A_{T \cup \ell}^*A_{T\cup \ell}- I$.
Therefore, $\|A_S^* a_\ell\|_2 \leq \|A_{T \cup \ell}^*
A_{T \cup \ell} - I\|_{2} \leq \delta$, and
\begin{align}
&\hspace*{-0.5in} \|A_S^\dagger a_\ell\|_2 \leq \|(A_S^*A_S)^{-1} A_S^* a_\ell\|_2\notag\\
&\leq \|(A_S^* A_S)^{-1}\|_{2} \|A_S^* a_\ell\|_2 \leq
\frac{\delta}{1-\delta}.\notag
\end{align}
The proof of (b) follows from the fact that $\|A_S^* a_\ell\|_2
\leq \mu_2(S)$. A similar estimate as above yields $\|A_S^\dagger
a_\ell\|_2 \leq (1-\delta)^{-1}\eta$.
\end{proof}

Proposition~\ref{lem_RIP} applies if $\delta_{k+1}$ is small while
in contrast Theorem~\ref{thm_RIP} works with $\delta_{2k}$, which
is generally larger than $\delta_{k+1}$. By \eqref{RIPn} the
condition $\delta_{k+1} \leq \delta$ can be satisfied if $n \geq
C_\delta k \log(N/k)$. Working with $\delta^*(S)$ instead of
$\delta_{k+1}$ allows to improve on the bound \eqref{RIPn} for
Gaussian, Bernoulli and random spherical matrices.


\begin{proposition}\label{prop:spherical} Let $S \subset \{1,\hdots,N\}$ be a set of cardinality $k$
and suppose that $A = \frac{1}{\sqrt{n}} \Phi \in \R^{n\times N}$,
where $\Phi$ is drawn at random according to a standard Gaussian
or Bernoulli distribution (with expectation $0$ and variance
$1/n$). Then $\delta^*(S) \leq \delta$ with probability at least
$1-\epsilon$
provided that
\begin{equation}\label{delta:star}
n \geq C_1 \delta^{-2} \max \{k \log(1/\delta), \log(N/\epsilon)\}
\end{equation}
for a suitable constant.

The same statement holds (with possibly a different constant) for
a random matrix whose columns are chosen independently at random
according to the uniform distribution on a sphere.
\end{proposition}
\begin{proof}
See Appendix~\ref{sec:deltas}.
\end{proof}
A straightforward extension of the proof, as in \cite{badadewa06},
also shows that a random matrix $A \in \R^{n\times N}$ with
independent columns drawn from the uniform distribution on the
sphere satisfies RIP, $\delta_k \leq \delta$ with probability at
least $1-\epsilon$ provided $n \geq C \delta^{-2} (k\log(N/k) +
\log(\epsilon^{-1}))$. Although this fact seems to be known, we are not aware
of reference where this is rigorously stated.



The next result relies on a theorem by Tropp \cite[Theorem
B]{tr06-2} that uses random support sets $S$ and allows to work
with the coherence $\mu$ alone. Note that choosing $S$ at random
is perfectly in line with an average-case analysis.
\begin{Theorem}\label{cor_cond2est}
Let $A \in \C^{n \times N}$ have unit norm columns and coherence
$\mu$. Let $S \subset \{1,\hdots,N\}$ be a set of cardinality
$k\geq 4$ chosen uniformly at random.
Let $\delta, \epsilon \in (0,1)$ and assume
that
\begin{align}\label{condmumu}
\mu^2 k \log(\epsilon^{-1}) & \leq c \delta^2,\\
\label{cond_normA}
 \frac{k}{N}\|A\|_{2}^2 & \leq \frac{\delta}{4e^{1/4}},
\end{align}
where $c = \frac{\log(2)e^{-1/2}}{4\cdot 144 \log(3)} \approx 6.64 \cdot 10^{-4}$.
Then
\begin{align}
\|A_S^\dagger a_\ell\|_2 \leq
\frac{\sqrt{c}\,\delta}{(1-\delta)\sqrt{\log(\epsilon^{-1})}}
\quad \mbox{ for all } \ell \notin S\notag
\end{align}
with probability at least $1-\epsilon$.
\end{Theorem}

\begin{proof}
The proof relies on \cite[Theorem 12]{tr06-2}. The formulation
below follows from \cite{tr06-2} by setting $s =
\log(\epsilon^{-1})/\log(k/2)$ and estimating
$\log(k/2+1)/\log(k/2) \leq \log(3)/\log(2)$ for $k\geq 4$.
\begin{Theorem}\label{thm_Sropp} Assume $A \in \C^{n\times N}$ has unit norm columns and
coherence $\mu$. Let $S \subset \{1,\hdots,N\}$ be a set of
cardinality $k \geq 4$ chosen uniformly at random. The condition
\begin{equation}\label{condmu3}
\sqrt{144 \log(3)\log(2)^{-1}\mu^2 k \log(\epsilon^{-1})} + \frac{k}{N}\|A\|^2_{2} \leq e^{-1/4} \delta
\end{equation}
implies
\[
\P(\|A_S^*A_S - I\| \geq \delta) \leq \epsilon.
\]
\end{Theorem}
Using (\ref{condmumu}) and the value of $c$, the square-root in
(\ref{condmu3}) becomes $\delta/(2e^{1/4})$. Combining this with
(\ref{cond_normA}) shows that (\ref{condmu3}) is satisfied.
Therefore, $\|A_S^* A_S - I\|_{2} \leq \delta$ with probability at
least $1-\epsilon$, which implies that
\[\|(A_S^* A_S)^{-1}\|_{2} \leq \frac{1}{1-\delta}.\] Finally,
\begin{align}
&\hspace*{-0.2in} \|A_S^\dagger a_\ell\|_2 \leq \|(A_S^*
A_S)\|_{2}
\|A_S^*a_\ell\|_2 \leq \frac{1}{1-\delta} \sqrt{k} \mu \notag\\
&\hspace*{0.3in} \leq
\frac{\sqrt{c}\,\delta}{(1-\delta)\sqrt{\log(\epsilon^{-1})}}\notag
\end{align}
by using condition (\ref{condmumu}) once more.
\end{proof}

\subsection{Comparison With Worst-Case Results}

Our average-case analysis depends on $\|A_S^\dagger a_\ell\|_2 $,
while the classical condition \eqref{cond1} of Proposition
\ref{prop:cond1} depends on $\|A_S^\dagger a_\ell\|_1$ and is
therefore significantly stronger.
Proposition~\ref{lem_RIP} establishes that the $2$-norm condition
can be satisfied as long as $\delta_{k+1}<1/2$. This is clearly
weaker than the worst case condition $\delta_{2k} < \sqrt{2} - 1
\approx 0.41$ of Proposition \ref{thm_RIP}.

Let us now compare worst-case and average results based on the
coherence $\mu$, by relying on Theorem~\ref{cor_cond2est}. For
simplicity, we consider the case in which $A$ is a unit-norm tight
frame, for which $\|A\|^2_{2} = \frac{N}{n}$. In this case,
\eqref{cond_normA} is equivalent to
$k \leq \frac{\delta}{4e^{1/4}}n$. If additionally $\mu = c
/\sqrt{n}$, then conditions \eqref{condmumu} and
\eqref{cond_normA} are both satisfied for fixed $\epsilon, \delta$
provided
\[
k \leq C'n.
\]
This beats the square-root bottleneck and even removes the
$\log$-factor present in estimates for the restricted isometry
constants, see \eqref{RIPn}. Moreover, we have the additional
advantage that the coherence is much easier to estimate than the
restricted isometry constants.


%
%
%


Combining Theorem~\ref{cor_cond2est} with the average-case
analysis of Theorems \ref{thm_average2} and \ref{thm_average}
shows that for a unit norm tight frame $A$ of coherence $\mu$
multichannel sparse recovery by \eqref{l1l2} can be ensured in the
average-case provided $k \leq C \mu^{-2}$, which can be as small
as $k \leq C n$. Moreover, the failure probability decays
exponentially in the number of channels.

In the next sections we provide further examples when we discuss
particular choices of the matrix $A$.

\section{Comparison with Multichannel Greedy Algorithms}
\label{sec:compare}

We now compare our results regarding $\ell_{2,1}$ optimization to
those obtained for the greedy algorithms $p$-thresholding and
$p$-SOMP \cite{grrascva07}. These are multichannel versions of
simple thresholding and orthogonal matching pursuit. For $1\leq p
\leq \infty$ they produce a $k$-sparse signal $\hat{X}$ from
measurements $Y = AX$ using a greedy search. To this end, we
improve slightly
on previous average-case performance results in \cite{grrascva07} for these
algorithms in the noiseless setting.

\subsection{Greedy Methods}

In $p$-thresholding, we select a set $S$ of $k$ indices whose
$p$-correlation with $Y$ are among the $k$ largest:
\begin{equation}
  \label{eq:DefPThresh}
  \|a_{\ell}^* Y\|_p \geq \|a_{j}^* Y\|_p,\quad
  \forall \ell \in S, \forall j \notin S.
\end{equation}
After the support $S$ is determined, the non-zero coefficients of
$\hat{X}$ are computed via an orthogonal projection: $\hat{X}^S =
A_S^\dagger Y$.

The $p$-SOMP algorithm is an iterative procedure. At each
iteration, an atom index $\ell_m$ is selected, and a residual is
updated. At the first iteration the residual is simply $Y_0 = Y$.
After $M$ iterations, the set of selected atoms being $S_M=
\{\ell_m\}_{m=1}^M$, the new residual is computed as $Y_M =
Y-A_{S_M} X_M = (I-P_{S_M}) Y$ where $X_M = A_{S_M}^\dagger Y$ and
$P_{S_M} = A_{S_M}A_{S_M}^\dagger$ is the orthogonal projection
onto the linear span of the selected atoms. The next selected atom
$k_{M+1}$ is the one which maximizes the $p$-correlation with the
residual $Y_M$,
\begin{equation}
  \label{eq:DefOMPSelect}
\|a_{\ell_{M+1}}^* Y_M\|_p = \max_{1 \leq \ell \leq N}
\|a_{\ell}^* Y_M\|_p.
\end{equation}

Using the probability model \eqref{prob:model}
average-case recovery theorems for $p$-thresholding and $p$-SOMP
have been proven in \cite[Theorems 4,6,7,8]{grrascva07,grmarascva07}.
We improve slightly on these in the following. (Note, however, that \cite{grrascva07} also treats the noisy case.) Our first result generalizes
the one in \cite{KV07} to the multichannel setup.

\begin{Theorem}\label{thm:thresh}
Let $A \in \C^{n \times N}$ have unit norm columns and local $2$-coherence
function $\mu_2(S)$ defined in \eqref{def:mu2}.
Let $X \in \R^{N\times L}$ with $\supp X \subset S$ where $S
\subset \{1,\hdots,N\}$, and such that the coefficients on $S$ are
given by \eqref{prob:model}, $X^S = \Sigma \Phi$, where we choose the real
spherical model for $\Phi$.
Set $Y =AX$ and $R = \max_{j} \sigma_j/\min_{j}\sigma_j$.
If
\begin{equation}\label{cond:thresh}
\theta = R \mu_2(S) < 1,
\end{equation}
then the probability that $2$-thresholding applied to $Y$ fails to recover $X$ is bounded
by
\[
N \exp\left(-L/2 (\theta^{-2} - \log(\theta^{-2}) -1 ) \right).
\]
If we use the complex spherical model instead of the real spherical
model then $L/2$ in the above probability estimate may be replaced by $L$.
\end{Theorem}

The probability bound of Theorem~\ref{thm:thresh} is similar to
that of Theorem~\ref{thm_average2}. However, in contrast to our
results for $\ell_{2,1}$-minimization, success of thresholding
suffers a dependency on the diagonal matrix $\Sigma$. The larger
the ratio $R$, the stronger the condition \eqref{cond:thresh} on
the maximal allowed sparsity $k$, and the larger the probability
of error.

\begin{proof} We proceed similarly as in \cite{KV07}. We denote by $\Theta$ the event that
$2$-thresholding fails. Clearly,
\begin{align}
&\hspace*{-0.2in} \P(\Theta) = \P(\min_{i \in S} \|a_i^* Y \|_2 <
\max_{\ell \notin
S} \|a_\ell^* Y\|_2 ) \notag\\
&\hspace*{0.1in}\leq \P(\min_{i \in S} \|a_i^* Y\|_2 < \rho ) +
\P(\max_{\ell \notin S} \|a_\ell^* Y\|_2 > \rho),\notag
\end{align}
where $\rho$ will be specified later. Denote by $Z_j$, $j \in S$, a
sequence of independent random vectors which are uniformly
distributed on the unit sphere of $\R^L$. Then,
\begin{equation}
\label{eq:term1} \P(\min_{i \in S} \|a_i^* Y\|_2 < \rho )=
\P\bl\min_{i \in S} \left\| \sum_{j \in S} a_i^* a_j \sigma_j
Z_j^*\right\|_2 < \rho\br.
\end{equation}
Now,
\begin{align}
&\hspace*{-0.1in}\left\|\sum_{j \in S}a_i^* a_j \sigma_j
Z_j^*\right\|_2 = \left\| \sigma_i Z_i^* + \sum_{j \in S, j\neq i}
a_i^* a_j
\sigma_j Z_j^*\right\|_2 \notag\\
& \hspace*{0.2in} \geq |\sigma_{\min}| - \left\| \sum_{j \in S,
j\neq i} \sigma_j \langle a_i, a_j\rangle Z_j^*\right\|_2.\notag
\end{align}
Substituting into (\ref{eq:term1}),
\begin{align}
&\P(\min_{i \in S} \|a_i^* Y\|_2 <\rho ) \notag\\
&\hspace*{0.2in}\leq \sum_{i \in S} \P\bl\left\|\sum_{j \in S,
j\neq i} \sigma_j a_i^* a_j Z_j^* \right\|_2 \geq \sigma_{\min} -
\rho \br.\notag
\end{align}
Choosing $\rho = \sigma_{\min}/2$ and applying Theorem
\ref{thm:Bernstein} we obtain
\begin{align}
&\P(\min_{i \in S} \|a_i^* Y\|_2 < \rho )\notag\\
& \hspace*{0.2in}\leq  k \exp(-L/2(\theta^{-2} - \log(\theta^{-2})
- 1))\notag
\end{align}
where we used the definition of $\theta$ and $\mu_2(S)$. 
Similarly we estimate
\begin{align}
&\P(\max_{\ell \notin S} \|a_\ell^* Y\|_2 > \sigma_{\min}/2)\notag\\
&\hspace*{0.2in}\leq (N-k) \exp(-L/2(\theta^{-2} -
\log(\theta^{-2}) -1)).\notag
\end{align}
Combining the two estimates completes the proof for the real case.
Choosing the vectors $Z_j$, $j\in S$, from the complex unit sphere
$S_\C^L$ and using Corollary \ref{cor:Bernstein} yields the
statement for the complex case.
\end{proof}


We now state the corresponding result for $2$-SOMP, which slightly improves the one in \cite{grrascva07} for the noiseless case. (Note that we restrict to $p=2$ here, although the theorem
is easily extended to general values of $p$.)

\begin{Theorem}\label{thm:OMP}
Let $A$ be a matrix with unit norm columns and constants
$\delta(S), \mu_2(S) < 1$ where $S \subset
\{1,\hdots, N \}$. Assume that
\begin{equation}\label{bound:SOMP}
\frac{\mu_2(S)^2 + (1+\epsilon)(1-\epsilon)^{-1} \mu_2(S)}{1-\delta(S)} \leq 1
\end{equation}
for some $\epsilon \in (0,1)$. Let $X$ be a random coefficient matrix
with support $S$ that is selected according to the real
Gaussian probability model, see \eqref{prob:model},
and let $Y = AX$. Then
$2$-SOMP applied to $Y$ recovers $X$ in $k$ steps with probability
at least
\begin{equation}\label{prob:boundOMP}
1 - N 2^k \exp(-\epsilon^2 A_L^2), 
\end{equation}
where $A_L \sim \sqrt{L}$ is given by \eqref{def:AL}.

If we use the complex Gaussian model instead of the real Gaussian model
then the same conclusion holds with $L$ replaced by
$2L$ in \eqref{prob:boundOMP}.
\end{Theorem}
\begin{proof}
See Appendix~\ref{sec:somp}.
\end{proof}
\begin{remark}\label{rem:OMP}
\begin{itemize}
\item[(a)] Due to the factor $2^k$ the probability bound
(\ref{prob:boundOMP}) becomes effective only when the number of
channels becomes comparable to the sparsity $k$. This drawback is
very likely due to the analysis and is not observed in practice.
However, it seems to be very difficult to remove this factor by a
more sophisticated proof technique.
\item[(b)] We require $\epsilon<1$, so that the probability
decay of \eqref{prob:boundOMP} is potentially slower than that
given by Theorem~\ref{thm_average2}.
\item[(c)] With $\delta = \epsilon = 1/2$ condition (\ref{bound:SOMP}) is satisfied if
$\mu_2(\Lambda) \leq 1/7$ while the probability estimate
\eqref{prob:boundOMP} behaves like $1- N 2^k \exp(-L/4)$.
\item[(d)] With the estimates $\delta(S) \leq \delta^*(S)$ and
$\mu_2(S) \leq \delta^*(S)$,
\eqref{bound:SOMP} with $\epsilon = 3/11$ is implied by
\[
\delta^*(S) < 1/3.
\]
\item[(e)] By Proposition \ref{lem_RIP} the condition $\delta^*(S) < 1/3$
 implies
$\|A_S^\dagger a_\ell\|_2 \leq 1/2$ for all $\ell \notin S$, i.e.,
the bounded norm condition ( \ref{cond:boundnorm}) of the average
case recovery result for mixed $\ell_{2,1}$. In other words, the
condition in (d) for SOMP is slightly stronger than the one for
$\ell_{2,1}$.
\end{itemize}
\end{remark}

\subsection{Comparison}

We now compare the average-case recovery conditions for mixed
$\ell_{2,1}$, thresholding and SOMP for the following choices of
the matrix $A$ which we will also use in the numerical
experiments:
\begin{enumerate}
\item Random spherical ensemble;
 \item Union of Dirac and Fourier;
 \item Time-Frequency shifts of the Alltop window.
\end{enumerate}

\subsubsection{Random spherical ensemble}

Assume that the random columns of
$A \in \R^{n \times N}$ are independent and uniformly distributed
on the sphere $S^{n-1}$. Let $S$ be a support set of size $k$.
Then according to Proposition \ref{lem_RIP}
 the condition
$\|A_S^\dagger a_\ell\|_2 \leq \alpha<1$ of Theorem
\ref{thm_average2} is implied by $\delta^*(S) \leq
\frac{\alpha}{1+\alpha} < 1/2$, while by Proposition
\ref{prop:spherical} the latter holds with probability at least
$1-\epsilon$ provided
\begin{equation}\label{nk:l1l2}
n \geq \max\{ C_1(\alpha) k, C_2(\alpha) \log(N/\epsilon)\}.
\end{equation}
Assuming, for example, $\alpha=1/4$, under the probability model
\eqref{prob:model}, the probability that reconstruction by
$\ell_{2,1}$ fails is bounded from above by $N \exp(-L/2(15 -
\log(16))) + \epsilon = N \exp( - c L) + \epsilon$ with $c \approx
6.1137$.

We now compare this result with the condition of Theorem
\ref{thm:thresh} concerning thresholding. As noted in
\eqref{mudelta:est}, $\mu_2(S) \leq \delta^*(S)$. Therefore, by
Proposition \ref{prop:spherical} we have
\[
\theta = 2 R \mu_2(S) \leq 2R \delta^*(S) < 1
\]
with probability at least $1-\epsilon$ provided
\begin{equation}\label{n:thresh}
n \geq C \frac{R^2}{\theta^2} \max\left\{k \log(R/\theta),
\log(N/\epsilon) \right\}
\end{equation}
and the failure probability of thresholding is bounded by
$N\exp( - L/2 (\theta^{-2} - \log(\theta^{-2})-1)) + \epsilon$.

Let us finally consider Theorem \ref{thm:OMP} for SOMP. By Proposition \ref{prop:spherical} the condition
$\delta^*(S) < 1/3$ in Remark \ref{rem:OMP} is satisfied with probability at least
$1-\epsilon$ provided
\begin{equation}\label{n:OMP}
n \geq \max\left\{C_1 k, C_2 \log(N/\epsilon)\right\}
\end{equation}
and the failure probability of SOMP is bounded by
\begin{equation}\label{prob:est:SOMP}
N2^k \exp(-9/121\, A_L^2) + \epsilon 
\end{equation}
with $A_L^2 \sim L$ if the real Gaussian probability model is used.

Conditions \eqref{nk:l1l2}, \eqref{n:thresh}, \eqref{n:OMP} for
$\ell_{2,1}$, thresholding and SOMP are rather similar. However,
condition \eqref{n:thresh} for thresholding
involves the ratio $R$. 
If $R$ is large then thresholding behaves much worse compared to
$\ell_{2,1}$ and SOMP. The probability estimate
\eqref{prob:est:SOMP} is the worst compared to the other two
algorithms due to the factor $2^k$. Therefore, $\ell_{2,1}$ gives the best known
theoretical average case result.


\subsubsection{Union of Dirac and Fourier}

Consider the $n \times 2n$ matrix $A = (I | F)$, where $I$ is the $n\times n$
identity matrix and $F$ is the normalized $n \times n$ Fourier
matrix.
The coherence of $A$ is easily
seen to be $\mu = 1/\sqrt{n}$. By Proposition \ref{prop:bound:mu}
condition \eqref{cond:boundnorm}, $\|A_S^\dagger a_\ell\|_2 \leq
\alpha$ with $\alpha = 1/2$ is satisfied for all support sets $S$
of cardinality at most $k$ provided
$$
\sqrt{\frac{k}{n}} + \frac{k-1}{2\sqrt{n}} < \frac{1}{2}.
$$
If $S$ is chosen at random then a much better bound (up to
constants) is obtained using Theorem~\ref{cor_cond2est}. In our
special case, however, further improvement is possible. A
reformulation of a result of \cite{caro06-1}, see also
\cite[Proposition 3]{tr06-2} shows the following. If the support
$S$ consists of $k_1$ arbitrary elements of $\{1,\hdots,n\}$ and
$k_2$ random elements of $\{n+1,\hdots,2n\}$ then with probability
at least $1-\epsilon$ we have $\delta(S) \leq 1/2$
provided
\begin{equation}\label{cond:DiracFourier}
k = k_1 + k_2 \leq \frac{c n}{\sqrt{\log\left(\frac{\epsilon}{Cn}\right) + \log(n)}},
\end{equation}
with $c = 0.25$. In particular $k \leq n/4$ and the same reasoning
as in the proof of Theorem~\ref{cor_cond2est} yields
\[
\|A_S^\dagger a_\ell\|_2 \leq \alpha = 1/2.
\]
Using one of the complex probability models in
Theorem \ref{thm_average2}, the failure
probability of $\ell_{2,1}$-minimization is bounded by
$N\exp(-L(4-\log(4)-1)) = N \exp(-cL)$ with $c \approx 1.61$.

To compute the performance of thresholding, note that condition
\eqref{cond:thresh}, $2 R \mu_2(S) \leq 2 R \mu \sqrt{k} \leq
\theta < 1, $ is satisfied provided
\begin{equation}\label{n:thresh-f}
n \geq \frac{4 R^2}{\theta^2}k.
\end{equation}
Assuming that the non-zero rows of the matrix $\Phi$ in the
probability model \eqref{prob:model} on the coefficients are
independent and uniformly distributed on the complex unit sphere
$S_\C^{L-1}$, the failure probability of thresholding is bounded
by $N \exp(-L(\theta^{-2} - \log(\theta^{-2}) - 1))$.

Assuming $\delta(S) \leq \delta = 1/2$ and $\mu \sqrt{k} \leq 1/7$, i.e.,
\begin{equation}\label{n:OMP-f}
n \geq 49\, k,
\end{equation}
the condition of Remark \ref{rem:OMP}(c) is satisfied since by
(\ref{mudelta:est}), $\mu_2(S) \leq \mu \sqrt{k} \leq 1/7$.
Then by Theorem \ref{thm:OMP} SOMP fails with probability at most
$N 2^k\exp(-A_{2L}^2/4)$ 
assuming the complex Gaussian
probability model. Assuming as in the discussion of $\ell_{2,1}$
that the support set is such that $k_1$ arbitrary elements of
$\{1,\hdots,n\}$ and $k_2$ random elements of $\{n+1,\hdots,2n\}$
are chosen with $k=k_1 + k_2$ then the assumed condition
$\delta(S) \leq 1/2$ is true with probability at least
$1-\epsilon$ provided \eqref{cond:DiracFourier} holds.

Similar conclusions on the comparison of the three algorithms as
in the previous example apply. We note, however, that
in contrast to $\ell_{2,1}$ and SOMP, the performance bound for thresholding
does not require a probability model on the support set $S$.




\subsubsection{Time-Frequency shifts of Alltop window}

Let $n \geq 5$ be a
prime. Denote by $(T_r g)_\ell = g_{\ell - r \mod n}$ and $(M_s
g)_\ell = e^{2\pi i s\ell/n} g_\ell$ the cyclic shift and
modulation operator, respectively. Then $T_r M_s$, $r,s =
0,\hdots,n-1$ forms the set of time-frequency shifts. Let $g_\ell
= \frac{1}{\sqrt{n}}e^{2\pi i \ell^3/n}$ be the so-called Alltop
window. Then define $A$ to be the $n \times n^2$ matrix with
columns being the time-frequency shifts $T_r M_s g$, $r,s =
0,\hdots, n-1$. The coherence of $A$ is $\mu = 1/\sqrt{n}$
\cite{hest03}.

As in the Fourier-Dirac case, under condition \eqref{n:thresh-f}
and the complex probability model of Theorem \ref{thm:thresh},
thresholding fails with probability at most $N \exp(-L(\theta^{-2}
- \log(\theta^{-2}) - 1))$.

For the analysis of $\ell_{2,1}$ and SOMP we assume that the
support $S$ is chosen
uniformly at random. 
As $A$ is the union of $n$ orthonormal bases we have $\|A\|_2^2 = n$.
Then choosing $\delta = 3/4$ in Theorem \ref{cor_cond2est} yields that under the condition
\[
n \geq C k \log(\epsilon^{-1})
\]
with a constant $C$ (which also implies \eqref{cond_normA}) we
have
\[
\|A_S^\dagger a_\ell\|_2 \leq 3\sqrt{c} \log^{-1/2}(\epsilon^{-1}) \leq \alpha \quad \mbox{ for all } \ell \notin S
\]
with probability at least $1-\epsilon$ where $\alpha = 3 \sqrt{c}
\approx 0.0773$. By Theorem \ref{thm_average2}, using one of the complex probability models,  the
failure probability of $\ell_{2,1}$ is then bounded by $N
\exp(-c_2 L) + \epsilon$ with $c_2 = \alpha^{-2} -
\log(\alpha^{-2}) - 1$.

For the analysis of SOMP we choose $\delta = 1/2$ in Theorem
\ref{thm_Sropp}. Assuming that the square-root in \eqref{condmu3}
is less than $\frac{9}{10} e^{-1/4} \frac{1}{2}$ is equivalent to
\begin{equation}\label{n:OMP:Alltop}
n \geq C k \log(\epsilon^{-1})
\end{equation}
with an appropriate $C$, and condition \eqref{condmu3} is
satisfied. Then with probability at least $1-\epsilon$ we have
$\delta^*(S) \leq 1/2$. Furthermore, as suggested by Remark
\ref{rem:OMP}(b) the condition $\mu_2(S) \leq 1/12$ is also
implied by \eqref{n:OMP:Alltop} since $\mu_2(S) \leq \sqrt{k} \mu
=\sqrt{\frac{k}{n}}$. Assuming the complex Gaussian probability model on
the non-zero coefficients of $X$ the failure probability of SOMP
is bounded by $N 2^k \exp(-A_{2L}^2/2) + \epsilon$ due to Theorem
\ref{thm:OMP}.


\section{Numerical Simulations}
\label{sec:sim}

We tested the three algorithms $\ell_{2,1}$ minimization,
thresholding and SOMP using the three different types of matrices
indicated in the previous section. The support set $S$ of the
sparse coefficient matrices $X$ was always selected uniformly at
random while the non-zero coefficients were selected at random
using one of the following choices of the probability model
\eqref{prob:model}, $X^S = \Sigma \Phi$:
\begin{enumerate}
\item $\Phi$ is chosen to be a real Gaussian random matrix (i.e.,
all entries independent and standard normally distributed);
$\Sigma$ has independent diagonal entries with standard normal
distribution. \item $\Phi$ is chosen to be a complex  Gaussian
random matrix (i.e., the real and imaginary parts of each entry
are chosen independently according to a standard normal
distribution); $\Sigma$ is equal to the identity.
\end{enumerate}
Note that $\Sigma = I$ is favorable for thresholding, while the
choice of $\Sigma$ should have no influence on the performance of
$\ell_{2,1}$ and only a mild influence on SOMP.


In the following figures the results of various simulation runs
are plotted (we always used $100$ simulations for each choice of
parameters).

In Fig.~\ref{fig:use} we plot the results when choosing $A$ from a
random spherical ensemble of size $n=32$ columns and $N=256$ rows
for $L=1,2,4$. The matrix $X$ was generated according to model
(1). The improvement with increasing $L$ is clearly evident.
\begin{figure}[h]
\subfigure[]{%
\centerline{
\resizebox{!}{2in}{\includegraphics{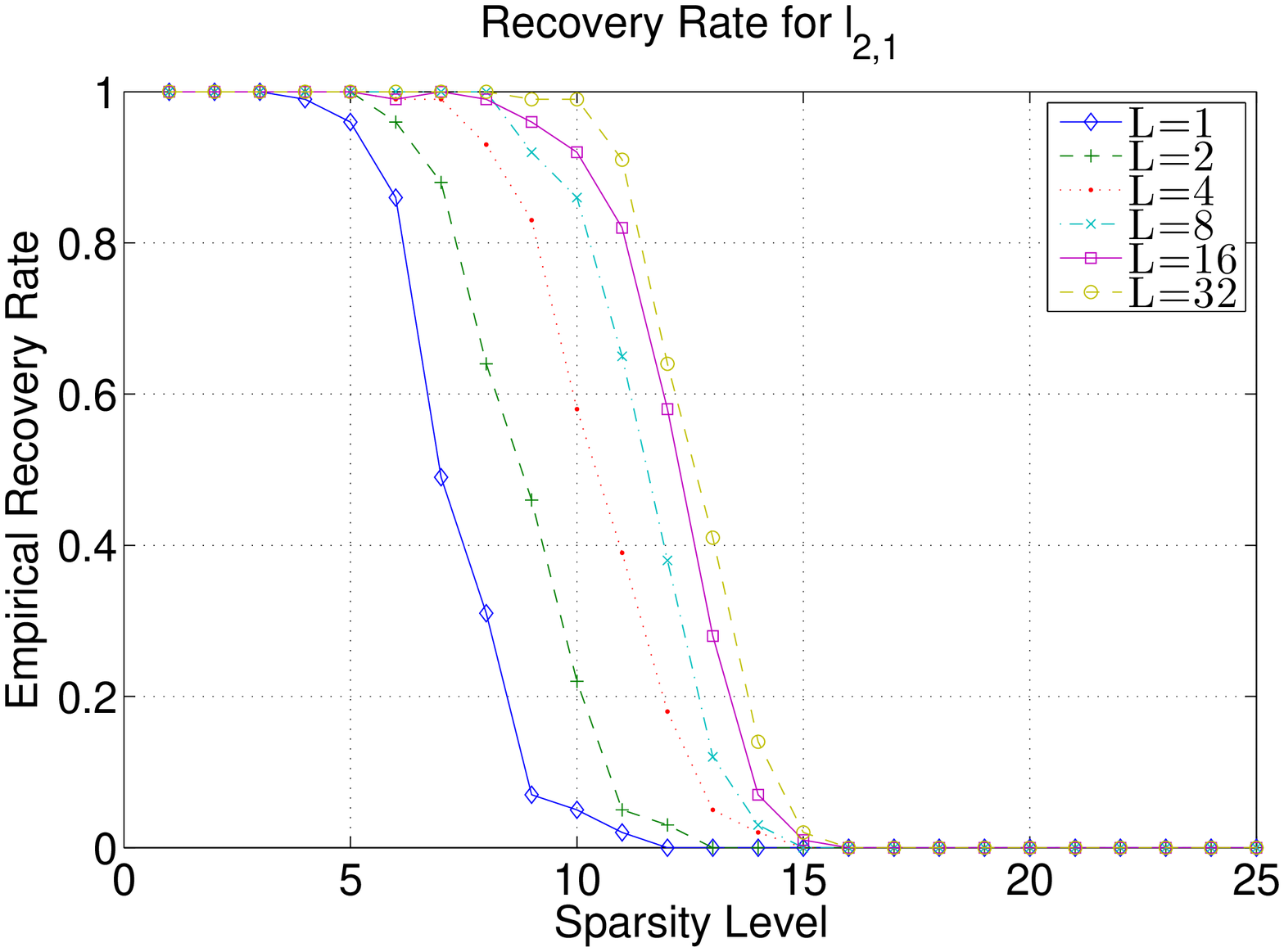}} }} \\
\subfigure[]{%
\centerline{
\resizebox{!}{2 in}{\includegraphics{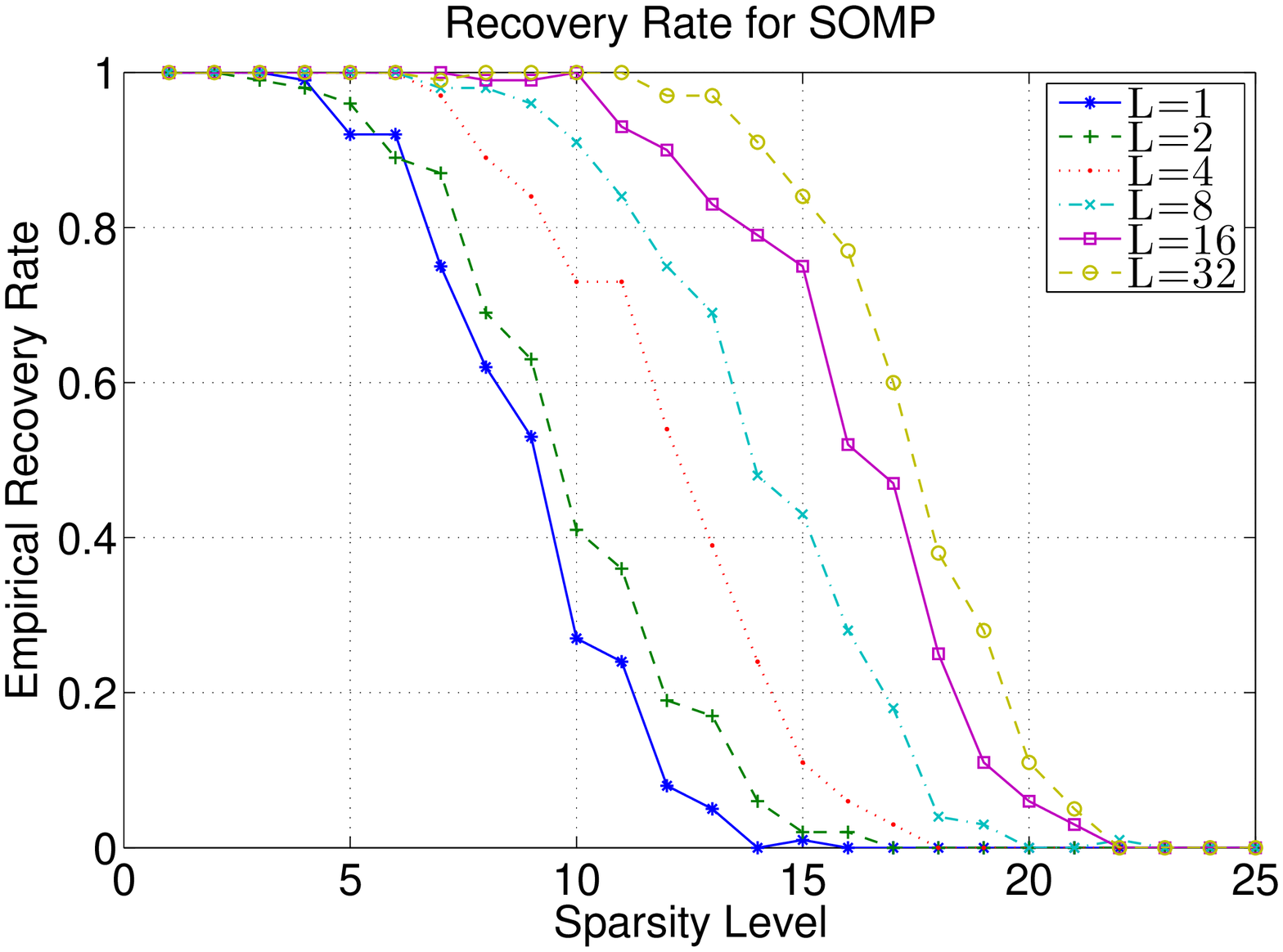}}}}\\ 
\subfigure[]{%
\centerline{
\resizebox{!}{2in}{\includegraphics{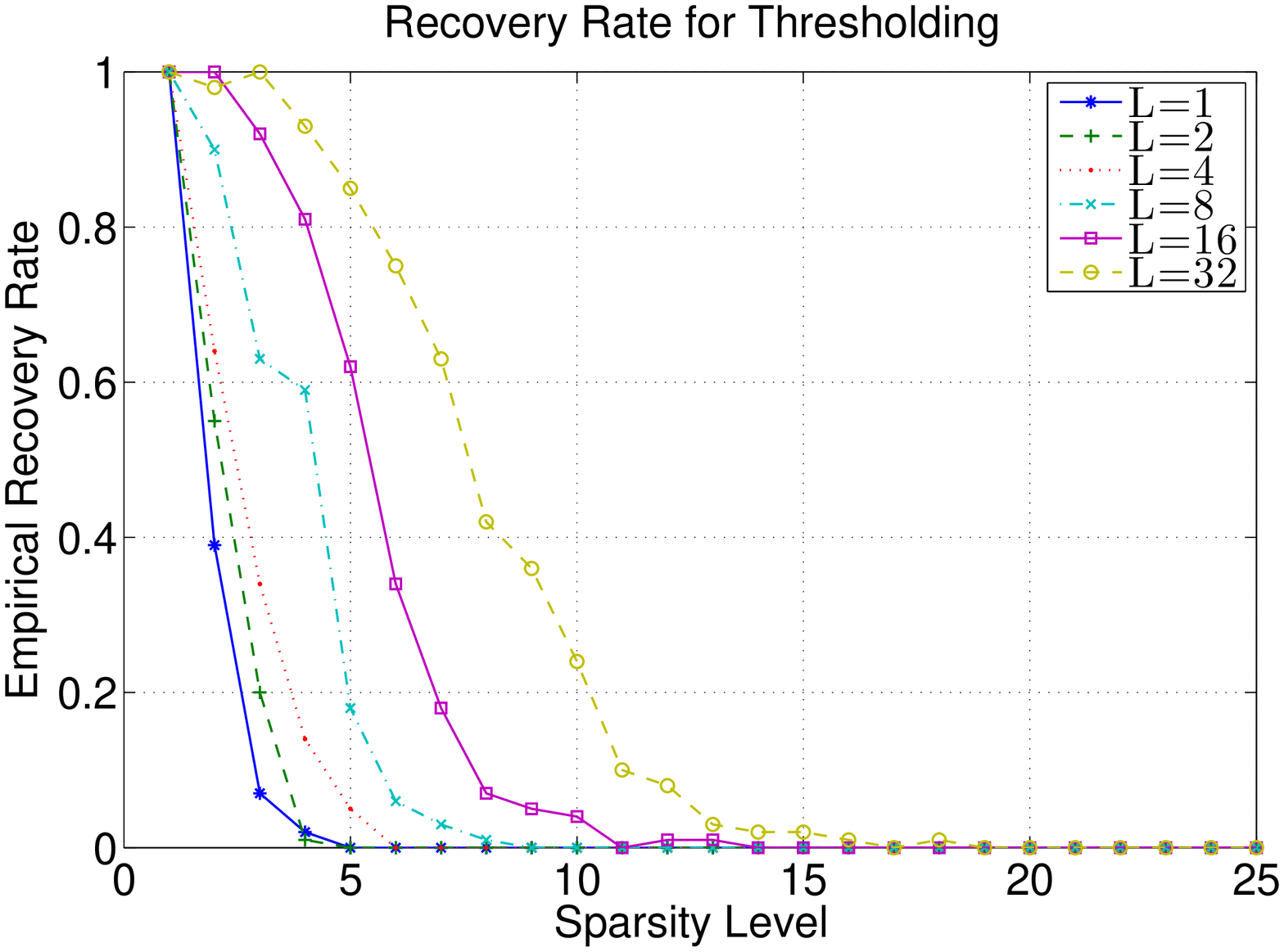}}} } \\
\caption{Multichannel recovery with $X$ generated according to
model (1) and $A$ chosen from a random spherical ensemble, (a)
$\ell_{2,1}$, (b) SOMP, (c) Thresholding.} \label{fig:use}
\end{figure}


In Fig.~\ref{fig:fd} we consider all three methods when $A$ is a
union of Dirac and Fourier bases, each with $32$ elements.
Therefore, $n=32$ and $N=64$. The matrix $X$ was generated
according to model (2). In this setting the performance using
thresholding is reasonable, though still worse than $\ell_{2,1}$
and SOMP.
\begin{figure}[h]
\subfigure[]{%
\centerline{
\resizebox{!}{2in}{\includegraphics{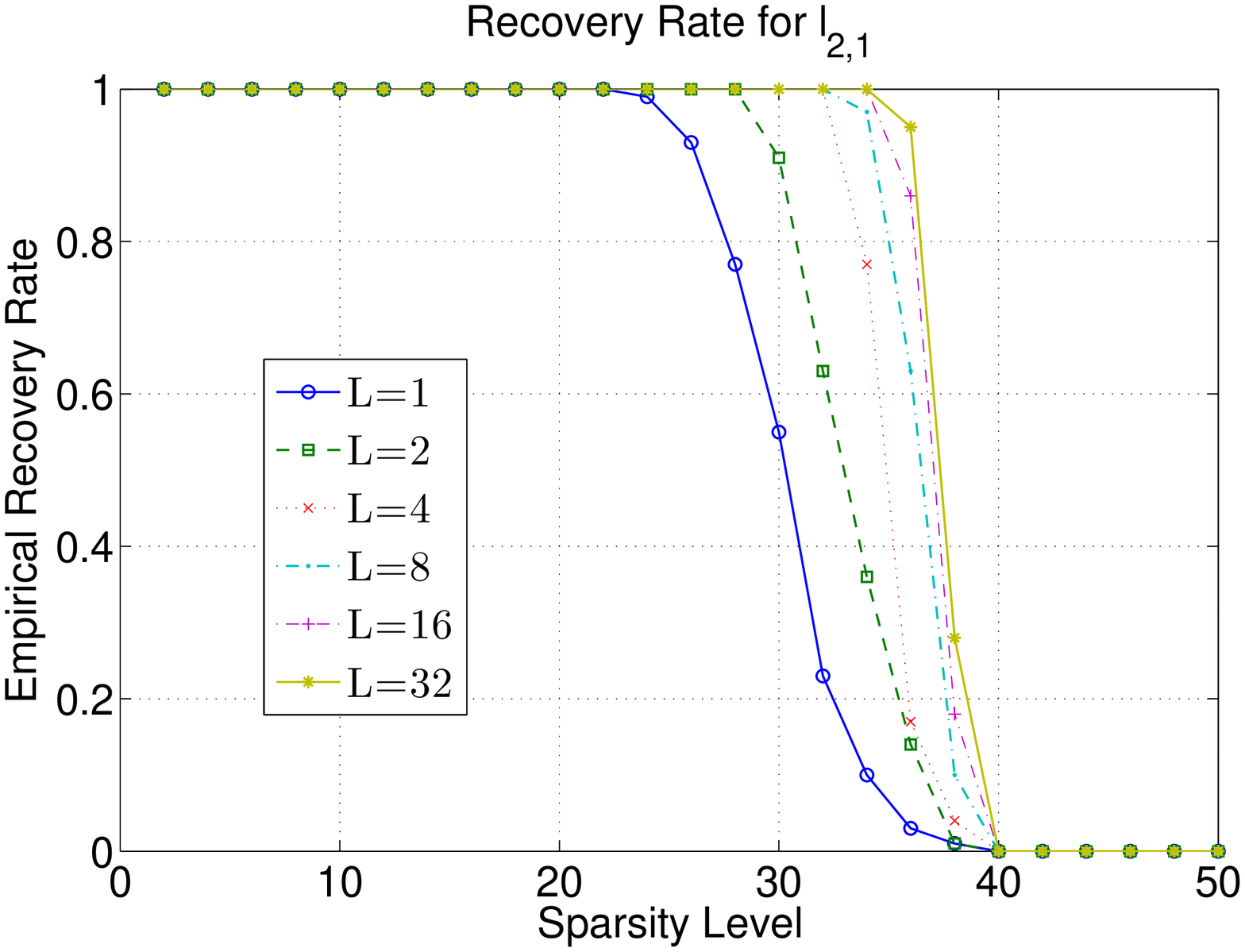}}} } \\
\subfigure[]{%
\centerline{
\resizebox{!}{2 in}{\includegraphics{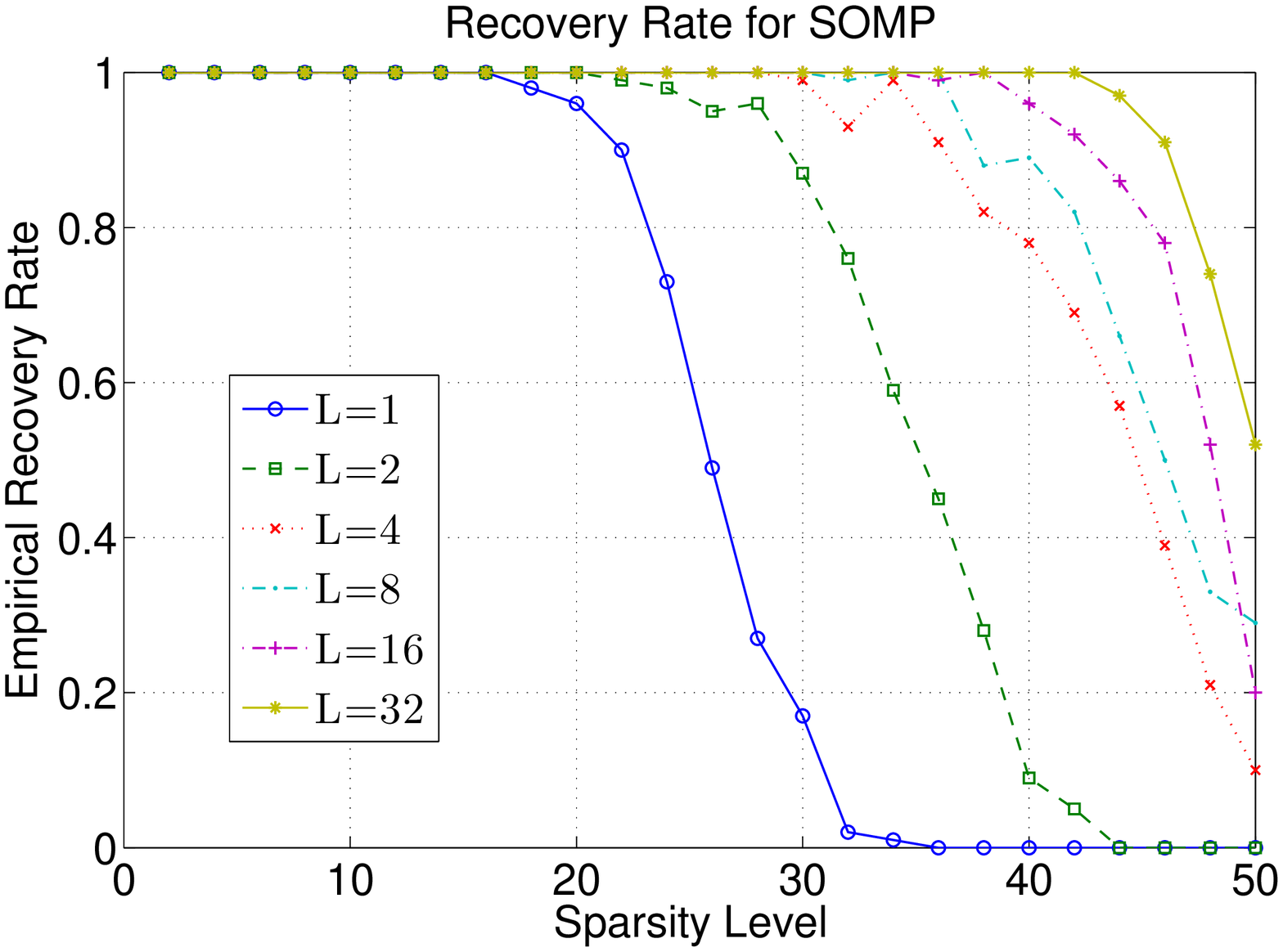}} }}\\
\subfigure[]{%
\centerline{ \resizebox{!}{2in}{\includegraphics{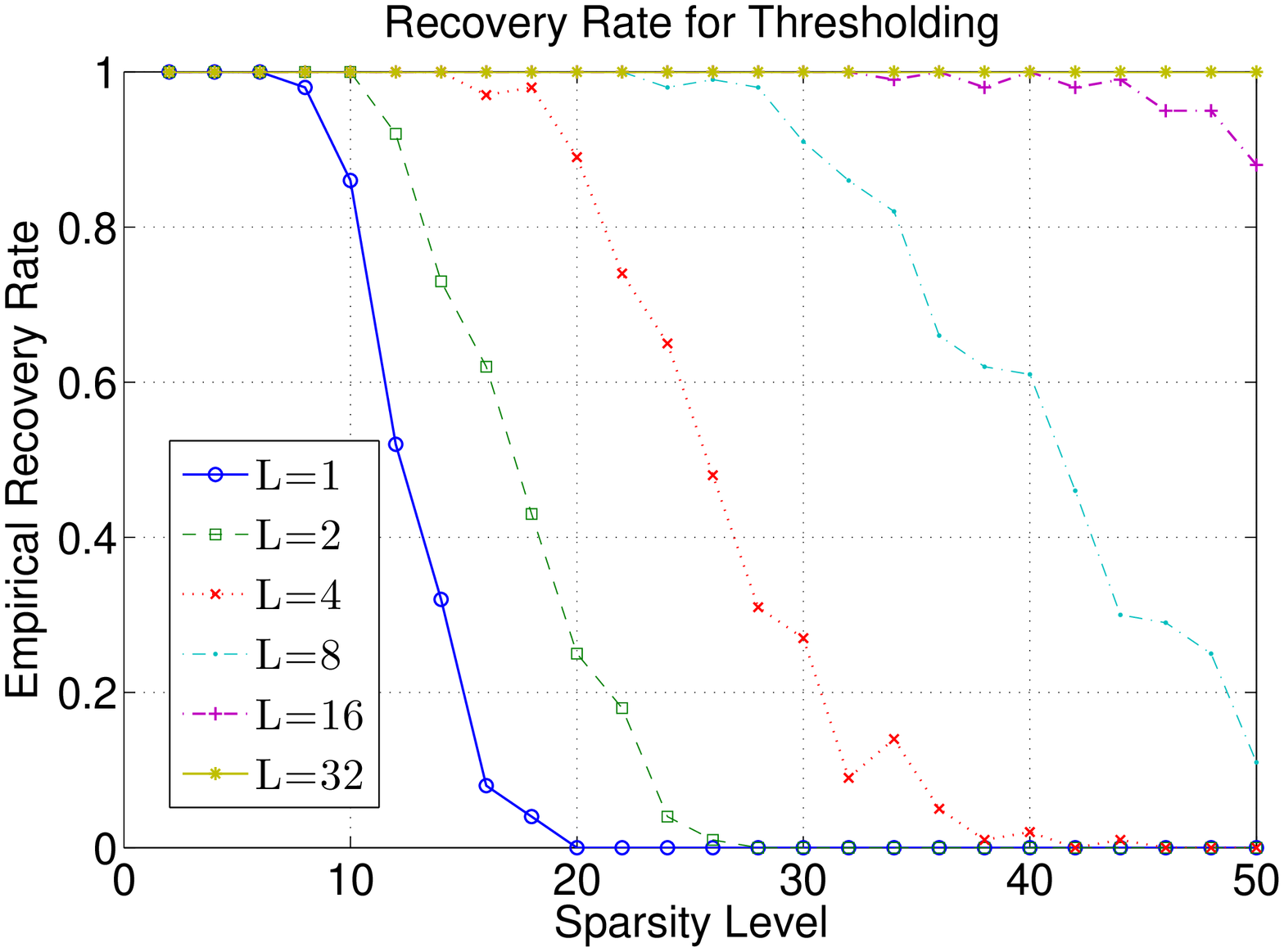}} }}
\caption{Multichannel recovery with $X$ generated according to
model (2) and $A$ a union of the Dirac and Fourier bases, (a)
$\ell_{2,1}$, (b) SOMP, (c) Thresholding.}\label{fig:fd}
\end{figure}


Finally, in Fig.~\ref{fig:gabor} we plot the results when using
time-frequency shifts of the Alltop window with $n=29$ and
$N=29^2=841$. Here the results of thresholding are extremely
poor and therefore not plotted.
\begin{figure}[h]
\subfigure[]{%
\centerline{
\resizebox{!}{2in}{\includegraphics{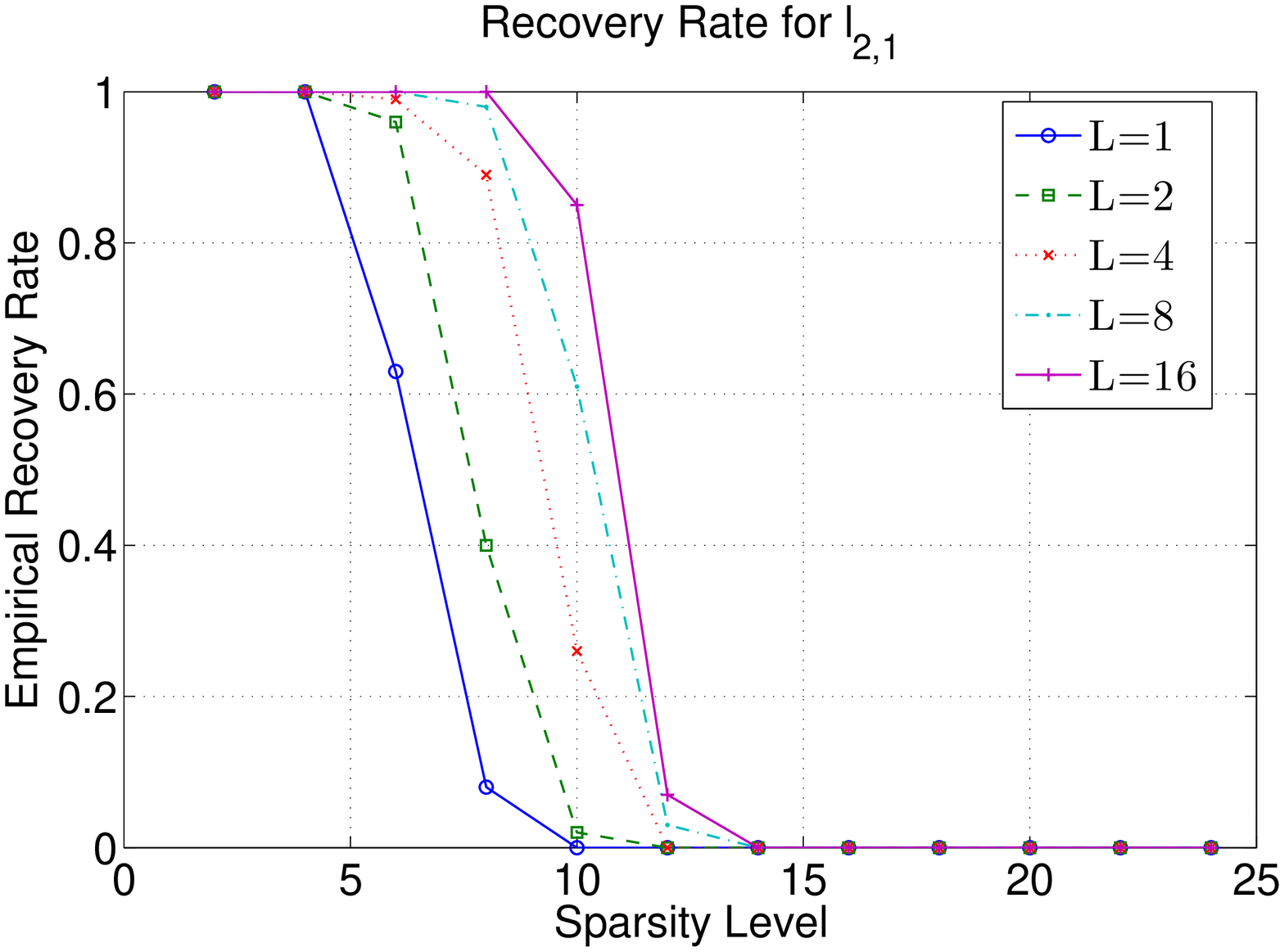}}} }\\
\subfigure[]{%
\centerline{ \resizebox{!}{2 in}{\includegraphics{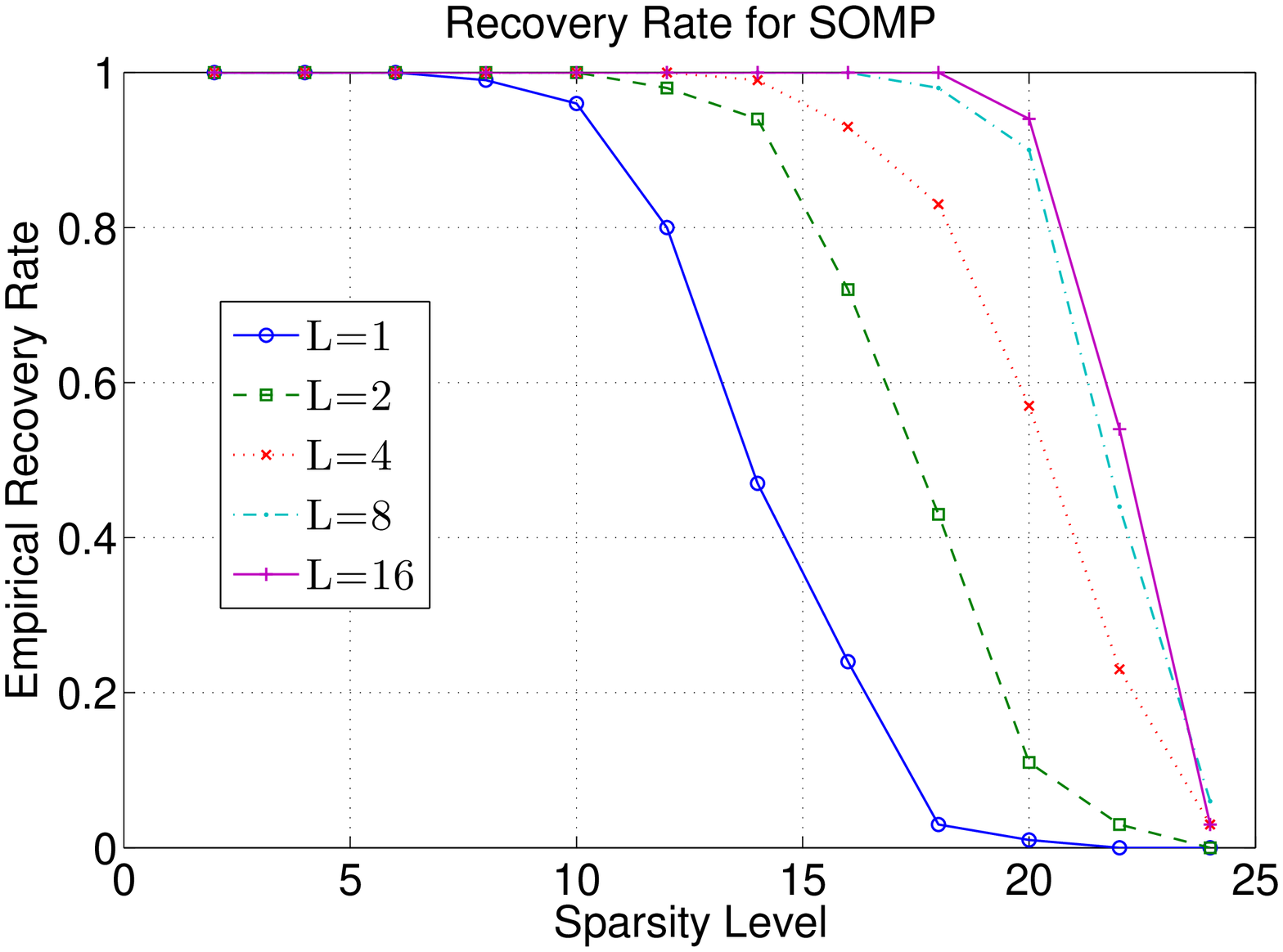}} } }
\caption{Multichannel recovery with $X$ generated according to
model (2) and $A$ chosen as time-frequency shifts of the Alltop
function (a) $\ell_{2,1}$ (b) SOMP.} \label{fig:gabor}
\end{figure}


In all three cases, SOMP performs better than the $\ell_{2,1}$
approach. However, both show clear performance advantage
with increasing $L$.

\section{Conclusion}
\label{sec:conc}

In this paper we analyzed the average-case performance of
$\ell_{2,1}$ recovery of multichannel signals. Our main result is
that under mild conditions on the sparsity and measurement matrix,
the probability of failure decays exponentially with the number of
channels. To develop this result we assumed a probability model on
the non-zero coefficients of a jointly sparse signal. The results
we obtained appear to be the best-known theoretical results on
multichannel recovery. Using the tools we developed for analyzing
the $\ell_{2,1}$ approach, we also improved slightly on previous
performance bounds for thresholding and SOMP.

 \appendices




\section{Proof of Theorem~\ref{thm:Bernstein}}
\label{sec:prob}

The proof uses the following extension of Khintchine's inequality
to higher dimensions stated in \cite{kwko01},
\[
\E \left\|\sum_{j=1}^k a_j Z_j \right\|_2^{p} \leq
\left(\frac{2}{L}\right)^{p/2}
\frac{\Gamma\left(\frac{L+p}{2}\right)}{\Gamma\left(\frac{L}{2}\right)}
\|a\|_2^p
\]
for all $p\geq 2$ and all vectors $a \in \R^k$. By splitting in real and imaginary parts
it easily follows that this inequality also holds for all $a\in \C^k$.
We may assume without loss of generality that
$\|a\|_2 = 1$. Then an application of Markov's inequality yields
\begin{align}
& \P\bl\left\|\sum_{j=1}^k a_j Z_j \right\|_2 \geq u\br \notag\\
&=
\P\left(\exp\bl\lambda L/2\left\|\sum_j a_j Z_j\right\|_2^2\br
\geq \exp(\lambda L u^2/2 )\right) \notag\\
& \leq \exp(- \lambda Lu^2 /2) \E\left[ \exp\bl\lambda L/2 \left\|\sum_j a_j Z_j\right\|_2^2\br\right] \notag\\
& = \exp(-\lambda L u^2/2) \sum_{i=0}^\infty (\lambda L/2)^i\E
\left\|\sum_{j=1}^k a_j Z_j\right\|_2^{2i}\notag\\
& \leq \exp(-\lambda L
u^2/2) \sum_{i=0}^\infty \lambda^i \frac{\Gamma(L/2+i)}{i!
\Gamma(L/2)}
\notag\\
& = \exp(-\lambda L u^2/2) \sum_{i=0}^\infty \frac{(L/2)_i}{i!}
\lambda^i \notag\\
&= \exp(-\lambda L u^2/2)
\frac{1}{(1-\lambda)^{L/2}},\label{prob1}
\end{align}
where $(a)_i = a(a+1)(a+2) \cdots (a+i-1)$ denotes the Pochhammer
symbol. The last equation is due to the fact that
$\sum_{i=0}^{\infty} \frac{(a)_i}{i!} \lambda^i$ is the Taylor
series of $(1-\lambda)^{-a}$, which converges for $\lambda < 1$.
Minimizing \eqref{prob1} with respect to $\lambda$ gives $\lambda
= 1-u^{-2}$. Inserting this value yields the statement of the
theorem.

\section{Proof of Proposition~\ref{prop:spherical} }
\label{sec:deltas}

Consider first the case of Gaussian or Bernoulli matrices.
According to Theorem 2.1 in \cite{rascva06} (see also Lemma 5.1 in
\cite{badadewa06}), we have $\|A_S^*A_S-I\|_{2} \geq \delta$ with
probability at most $2(1+12/\delta)^k \exp(-c_0/9 n \delta^2)$
with $c_0 = 7/18$. A similar estimate holds for $\|A_{S\cup\ell}^*
A_{S \cup \ell} - I\|_2$ with $\ell \notin S$. A union bound over
all $\ell \notin S$ yields $\delta^*(S) \geq \delta$ with
probability at most $2 N (1+12/\delta)^k \exp(-c_0/9 n \delta^2)$.
This term is less than $\epsilon$ if \eqref{delta:star} holds.

Now consider a random matrix $\Psi \in \R^{n \times N}$ with
independent columns that are uniformly distributed on the sphere
$S^{n-1}$. Then $\Psi$ has the same distribution as $D A$,
where $A$ is Gaussian matrix as above,
$D = \operatorname{diag}
(s_1^{-1},\hdots,s_N^{-1})$ and $s_j = \sqrt{n} \|\Phi_j\|_2$
where $\Phi_j \in \R^n$ is a vector of independent standard
normally-distributed random variables. We now use the following
measure concentration inequality \cite[Corollary (2.3)]{ba05} or
\cite[eq.\ (2.6)]{care08} for a standard Gaussian vector $Z \in
\R^n$,
\begin{align}
\P(\|Z\|_2^2 \geq \frac{n}{1-\gamma}) &\leq \exp (-\gamma^2 n/4),\notag\\
\quad \P(\|Z\|_2^2 \leq (1-\gamma)n) & \leq \exp(-\gamma^2 n /4).\notag
\end{align}
By a union bound this implies that
\begin{align}
&\P\left(1-\gamma \leq \min_{j=1,\hdots,N} s_j^2 \leq
\max_{j=1,\hdots,N} s_j^2 \leq
 \frac{1}{1-\gamma}\right)\notag\\
 &\hspace*{0.2in}\geq 1 - 2N \exp(-\gamma^2 n /4). \label{conc_ineq}
\end{align}
By the above reasoning, we have $(1-\delta/3)\|x\|_2^2 \leq
\|Ax\|_2 \leq (1+\delta/3)\|x\|_2^2$ for all $x$ with $\supp x
\subset S \cup \{\ell\}$ for some $\ell \notin S$ with probability
at least $1-\epsilon$ provided \eqref{delta:star} holds with a
suitable constant. If additionally $1-\gamma \leq
\min_{j=1,\hdots,N} s_j^2 \leq \max_{j=1,\hdots,N} s_j^2 \leq
\frac{1}{1-\gamma}$ for $\gamma = \delta/4$ then $(1-\delta)
\|x\|_2^2 \leq \|DAx\|_2^2 = \|\Psi x\|_2^2 \leq
(1+\delta)\|x\|_2^2$ for all $x$ with $\supp x \subset S \cup
\{\ell\}$ for some $\ell \notin S$. By a union bound and
\eqref{conc_ineq} this holds with probability at least
$1-2\epsilon$ provided \eqref{delta:star} holds and $2N
\exp(-\delta^2 n/64) \leq \epsilon$, the latter being equivalent
to $n \geq 64 \delta^2 \log(2N/\epsilon)$. Adjusting the constant
in \eqref{delta:star} completes the proof.

\section{Proof of Theorem~\ref{thm:OMP}}
 \label{sec:somp}

 We assume
that until a certain step SOMP has selected only correct indices,
collected in $J \subset S$. Let us first estimate the probability
that it selects a correct element of $S \setminus J$ also in the
next step.

We denote by $P_J = A_J A_J^\dagger$ the orthogonal projection
onto the span of the columns of $A$ in $J$, and $Q_J = I - P_J$.
The residual at the current iteration is given by $Y_M = Q_J Y =
Q_J A_S X = Q_J A_S \Sigma \Phi$. SOMP selects a correct index in
$S \setminus J$ in the next step if
\begin{equation}
\max_{\ell \in S \setminus J} \|a_\ell^* Q_J A_S \Sigma \Phi \|_2
> \max_{\ell \notin S} \|a_\ell^* Q_J A_S \Sigma \Phi\|_2.
\end{equation}
By Theorem 11 in \cite{grrascva07} (which is proven using Theorem
\ref{thm_conc}; note that there is a slight error in
\cite{grrascva07} in the computation of the constant $A_L$) we
have the following concentration of measure inequalities
\begin{align}
&\P\left( \max_{\ell \in S \setminus J} \|a_\ell^* Q_J A_S \Sigma \Phi\|_2
< (1+\epsilon) C_2(L) \times \right.\notag\\
& \phantom{\P\big(} \left.
\times \max_{\ell \in S \setminus J} \|a_\ell^* Q_J A_S \Sigma\|_2\right) \leq \exp(-\epsilon^2 A_L^2),\notag\\
&\P\left(\max_{\ell \notin S} \|a_\ell^* Q_J A_S \Sigma \Phi\|_2
 > (1-\epsilon) C_2(L) \times \right.\notag\\
& \phantom{\P\big(} \left.
\times  \max_{\ell \notin S} \|a_\ell^* Q_J A_S
\Sigma\|_2\right) \leq |S^c| \exp(-\epsilon^2 A_L^2),\notag
\end{align}
where $A_L$ is the constant in \eqref{def:AL} and $C_2(L) = \E \|
Z \|_2$ with $Z = (Z_1,\hdots, Z_L)$ being a vector of independent
standard normal variables. Now we assume that
\begin{align}
&(1+\epsilon) C_2(L)\max_{\ell \in S \setminus J} \|a_\ell^* Q_J
A_S \Sigma\|_2 \notag\\
&\geq  (1-\epsilon) C_2(L) \max_{\ell \notin S}
\|a_\ell^* Q_J A_S \Sigma\|_2. \label{max:cond}
\end{align}
Then by the above and a union bound the probability that SOMP
fails can be bounded by
\begin{align}
&\P(\max_{\ell \in S \setminus J} \|a_\ell^* Q_J A_S \Sigma \Phi
\|_2 \leq \max_{\ell \notin S} \|a_\ell^* Q_J A_S \Sigma \Phi\|_2)\notag\\
&\hspace*{0.2in}\leq (|S^c| + 1) \exp(-\epsilon^2 A_L^2).
\label{prob:bound:SOMP}
\end{align}
Let us consider now the maximum on the right hand side of
\eqref{max:cond}. First note that $P_J a_\ell = a_\ell$ for all
$\ell \in J$, in other words $Q_J a_\ell = 0$. Hence, we can
estimate
\begin{align}
&\hspace*{-0.2in} \max_{\ell \notin S} \|a_\ell^* Q_J A_S
\Sigma\|_2^2 = \max_{\ell \notin S}
\|\Sigma_{S\setminus J}  A_{S\setminus J}^* Q_J a_\ell\|_2^2\notag\\
&\leq \max_{\ell \notin S} \sum_{j \in S \setminus J} \sigma_j^2
|\langle Q_J a_j, a_\ell\rangle|^2 \notag\\
&\leq \max_{i \in S \setminus J}
\sigma_i^2 \max_{\ell \notin S} \sum_{j \in S \setminus J}
|\langle Q_J a_j, a_\ell\rangle|^2.    \notag
\end{align}
Furthermore, for $\ell \notin S$ we have
\begin{align}
&\hspace*{-0.2in}\left(\sum_{j \in S \setminus J} |\langle Q_J
a_j, a_\ell\rangle|^2\right)^{1/2} = \|A_{S\setminus J}^* Q_J
a_\ell\|_2\notag\\
&=\|A_{S \setminus J}^* (I - P_J) a_\ell\|_2\notag\\
&\leq \|A_{S\setminus J}^* a_\ell \|_2 + \|A^*_{S\setminus J} A_J
(A_J^* A_J)^{-1} A_J^* a_\ell\|_2 \notag\\
&\leq \mu_2(S \setminus J) +
\|A^*_{S \setminus J} A_J\|_{2} \|(A_J^*A_J)^{-1}\|_2
\|A_J^* a_\ell\|_2\notag\\
&\leq \mu_2(S) + \frac{\delta(S)}{1-\delta(S)} \mu_2(S) =
\frac{\mu_2(S)}{1-\delta(S)}, \notag
\end{align}
where we used the fact that $A_{S \setminus J}^* A_J$ is a
submatrix of $A_S^*A_S - I$.

Next we consider the maximum on the left hand side of
\eqref{max:cond}. We can estimate
\begin{align}
&\max_{\ell \in S\setminus J} \|a_\ell^* Q_J A_S \Sigma\|_2^2 =
\max_{\ell \in S \setminus J} \sum_{j \in S \setminus J}
\sigma_j^2 |\langle Q_J a_\ell, a_j\rangle|^2 \notag\\
&\hspace*{0.2in} \geq \max_{\ell \in S \setminus J} \sigma_\ell^2
\inf_{j \in S\setminus J} |\langle Q_J a_j,a_j\rangle|^2.\notag
\end{align}
Furthermore, for $j \in S \setminus J$
\begin{align}
&\hspace*{-0.7in}|\langle Q_J a_j, a_j\rangle| = |\langle (I-P_J)
a_j, a_j
\rangle| \notag\\
&= |1-a_j^* A_J (A_J^* A_J)^{-1} A_J^* a_j|\notag\\
&\geq 1 - \|A_J^* a_j\|^2 \|(A_J^*A_J)^{-1}\|_2\notag\\
&\geq 1-\mu_2(S)^2 (1-\delta(S))^{-1}.\notag
\end{align}
Combining the above estimates, condition \eqref{max:cond} is
satisfied if
\[
(1+\epsilon) \frac{\mu_2(S)}{1-\delta(S)} \geq (1-\epsilon)
\left(1-\frac{\mu_2(S)^2}{1-\delta(S)}\right),
\]
which is equivalent to \eqref{bound:SOMP}.

In order to complete the proof, we note that OMP successfully
recovers the correct signal if \eqref{max:cond} holds for all $J
\subset S$. By a union bound of \eqref{prob:bound:SOMP} over all
those $2^k$ subsets this is true with probability at least
$1 - N 2^k \exp(-\epsilon^2 A_L^2)$ 
provided condition \eqref{bound:SOMP} holds.

The extension to the complex valued case is straightforward.


\begin{thebibliography}{10}

\bibitem{badadewa06}
R.~G. {B}araniuk, M.~{D}avenport, R.~A. {D}e{V}ore, and
M.~{W}akin.
\newblock {A} simple proof of the restricted isometry property for random
  matrices.
\newblock {\em {C}onstr. {A}pprox.}, 28(3):253--263, 2008.

\bibitem{babadusawa05}
D.~{B}aron, M.~B. {W}akin, M.~F. {D}uarte, S.~{S}arvotham, and
R.~G.
  {B}araniuk.
\newblock {D}istributed compressed sensing.
\newblock {\em preprint}, 2005.

\bibitem{ba05}
A.~{B}arvinok.
\newblock Measure concentration, 2005.
\newblock lecture notes.

\bibitem{care08}
E.~{C}and{\'e}s and B.~{R}echt.
\newblock Exact matrix completion via convex optimization.
\newblock {\em preprint}, 2008.

\bibitem{caro06-1}
E.~{C}and\`es and J.~{R}omberg.
\newblock {Q}uantitative robust uncertainty principles and optimally sparse
  decompositions.
\newblock {\em {F}ound. {C}omput. {M}ath.}, 6(2):227--254, 2006.

\bibitem{ca08}
E.~J. {C}and{\`e}s.
\newblock {T}he restricted isometry property and its implications for
  compressed sensing.
\newblock {\em {C}. {R}. {A}cad. {S}ci. {P}aris {S}'er. {I} {M}ath.},
  346:589--592, 2008.

\bibitem{CRT06}
E.~J. Cand\`{e}s, J.~Romberg, and T.~Tao.
\newblock Robust uncertainty principles: Exact signal reconstruction from
  highly incomplete frequency information.
\newblock {\em IEEE Trans. Inform. Theory}, 52(2):489--509, Feb. 2006.

\bibitem{carota06-1}
E.~J. {C}and{\`e}s, J.~{R}omberg, and T.~{T}ao.
\newblock {S}table signal recovery from incomplete and inaccurate measurements.
\newblock {\em {C}omm. {P}ure {A}ppl. {M}ath.}, 59(8):1207--1223, 2006.

\bibitem{CT05}
E.~J. Cand\`{e}s and T.~Tao.
\newblock Decoding by linear programming.
\newblock {\em IEEE Trans. Inform. Theory}, 51(12):4203--4215, Dec. 2005.

\bibitem{cata06}
E.~J. {C}and{\`e}s and T.~{T}ao.
\newblock {N}ear optimal signal recovery from random projections: universal
  encoding strategies?
\newblock {\em {I}{E}{E}{E} {T}rans. {I}nform. {T}heory}, 52(12):5406--5425,
  2006.

\bibitem{Chen}
J.~Chen and X.~Huo.
\newblock Theoretical results on sparse representations of multiple-measurement
  vectors.
\newblock {\em IEEE Trans. Signal Processing}, 54(12):4634--4643, Dec. 2006.

\bibitem{chdosa99}
S.~S. {C}hen, D.~L. {D}onoho, and M.~A. {S}aunders.
\newblock {A}tomic decomposition by {B}asis {P}ursuit.
\newblock {\em {S}{I}{A}{M} {J}. {S}ci. {C}omput.}, 20(1):33--61, 1999.

\bibitem{Cotter}
S.~F. Cotter, B.~D. Rao, K.~Engan, and K.~Kreutz-Delgado.
\newblock Sparse solutions to linear inverse problems with multiple measurement
  vectors.
\newblock {\em IEEE Trans. Signal Processing}, 53(7):2477--2488, July 2005.

\bibitem{avdama97}
G.~{D}avis, S.~{M}allat, and M.~{A}vellaneda.
\newblock {A}daptive greedy approximations.
\newblock {\em {C}onstr. {A}pprox.}, 13(1):57--98, 1997.

\bibitem{do06-2}
D.~L. {D}onoho.
\newblock {C}ompressed sensing.
\newblock {\em {I}{E}{E}{E} {T}rans. {I}nform. {T}heory}, 52(4):1289--1306,
  2006.

\bibitem{do06}
D.~L. {D}onoho.
\newblock {F}or most large underdetermined systems of linear equations the
  minimal $l^1$ solution is also the sparsest solution.
\newblock {\em {C}ommun. {P}ure {A}ppl. {A}nal.}, 59(6):797--829, 2006.

\bibitem{DH01}
D.~L. Donoho and X.~Huo.
\newblock {Uncertainty principles and ideal atomic decomposition}.
\newblock {\em IEEE Transactions Info. Theory}, 47(7):2845--2862, 2001.

\bibitem{efhajoti04}
B.~{E}fron, T.~{H}astie, I.~{J}ohnstone, and R.~{T}ibshirani.
\newblock {L}east angle regression.
\newblock {\em {A}nn. {S}tatist.}, 32(2):407--499, 2004.

\bibitem{E08}
Y.~C. Eldar.
\newblock Compressed sensing of analog signals.
\newblock submitted to {\em IEEE Trans. Signal Processing}.

\bibitem{EB08}
Y.~C. Eldar and H.~B{\"o}lcskei.
\newblock Block-sparsity: Coherence and efficient recovery.
\newblock {\em {\em to appear in} ICASSP09}.

\bibitem{EM082}
Y.~C. Eldar and M.~Mishali.
\newblock Robust recovery of signals from a union of subspaces.
\newblock {\em {\em submitted to } IEEE Trans. Inf. Theory}.

\bibitem{fora06}
M.~{F}ornasier and H.~{R}auhut.
\newblock {R}ecovery algorithms for vector valued data with joint sparsity
  constraints.
\newblock {\em {S}{I}{A}{M} {J}. {N}umer. {A}nal.}, 46(2):577--613, 2008.

\bibitem{fu04}
J.~J. {F}uchs.
\newblock {O}n sparse representations in arbitrary redundant bases.
\newblock {\em {I}{E}{E}{E} {T}rans. {I}nform. {T}heory}, 50(6):1341--1344,
  2004.

\bibitem{grmarascva07}
R.~{G}ribonval, B.~{M}ailhe, H.~{R}auhut, K.~{S}chnass, and
P.~{V}andergheynst.
\newblock {A}verage case analysis of multichannel thresholding.
\newblock In {\em {P}roc. {I}{E}{E}{E} {I}ntl. {C}onf. {A}coust. {S}peech
  {S}ignal {P}rocess.}, 2007.

\bibitem{grrascva07}
R.~{G}ribonval, H.~{R}auhut, K.~{S}chnass, and P.~{V}andergheynst.
\newblock {A}toms of all channels, unite! {A}verage case analysis of
  multi-channel sparse recovery using greedy algorithms.
\newblock {\em {J}. {F}ourier {A}nal. {A}ppl.}, 14(5):655--687, 2008.

\bibitem{bogokikslu07}
S.~{K}im, K.~{K}sh, M.~{L}ustig, S.~{B}oyd, and D.~{G}orinevsky.
\newblock {A} method for large-scale l1-regularized least squares problems with
  applications in signal processing and statistics.
\newblock {\em {I}{E}{E}{E} {J}. {S}el. {T}op. {S}ignal {P}roces.},
  4(1):606--617, 2007.

\bibitem{kwko01}
H.~{K}{\"o}nig and S.~{K}wapie{\'n}.
\newblock {B}est {{K}}hintchine type inequalities for sums of independent,
  rotationally invariant random vectors.
\newblock {\em {P}ositivity}, 5(2):115--152, 2001.

\bibitem{le01}
M.~{L}edoux.
\newblock {\em {T}he {C}oncentration of {M}easure {P}henomenon}.
\newblock {A}{M}{S}, 2001.

\bibitem{leta91}
M.~{L}edoux and M.~{T}alagrand.
\newblock {\em {P}robability in {B}anach spaces. {I}soperimetry and
  processes.}, volume~23.
\newblock {S}pringer-{V}erlag, {B}erlin, {H}eidelberg, {N}ew{Y}ork, 1991.

\bibitem{ME07}
M.~Mishali and Y.~C. Eldar.
\newblock Blind multi-band signal reconstruction: Compressed sensing for analog
  signals.
\newblock {\em CCIT Report no. 639, EE Dept., Technion - Israel Institute of
  Technology; {\em submitted to } IEEE Trans. Signal Process.}, Sep. 2007.

\bibitem{ME08}
M.~Mishali and Y.~C. Eldar.
\newblock Reduce and boost: Recovering arbitrary sets of jointly sparse
  vectors.
\newblock {\em IEEE Trans. Signal Process.}, 56(10):4692--4702, Oct. 2008.

\bibitem{ME09}
M.~Mishali and Y.~C. Eldar.
\newblock From theory to practice: Sub-{N}yquist sampling of sparse wideband
  analog signals.
\newblock {\em arXiv 0902.4291; {\em submitted to } IEEE Selected Topics on
  Signal Process.}, 2009.

\bibitem{NT08}
D.~Needell and J.~A. Tropp.
\newblock {CoSaMP}: Iterative signal recovery from incomplete and inaccurate
  samples.
\newblock {\em {\em submitted}}, 2008.

\bibitem{NV08}
D.~Needell and R.~Vershynin.
\newblock Signal recovery from incomplete and inaccurate measurements via
  regularized orthogonal matching pursuit.
\newblock {\em {\em submitted}}, 2008.

\bibitem{ra05-7}
H.~{R}auhut.
\newblock {R}andom sampling of sparse trigonometric polynomials.
\newblock {\em {A}ppl. {C}omput. {H}armon. {A}nal.}, 22(1):16--42, 2007.

\bibitem{ra07}
H.~{R}auhut.
\newblock {O}n the impossibility of uniform sparse reconstruction using greedy
  methods.
\newblock {\em {S}ampl. {T}heory {S}ignal {I}mage {P}rocess.}, 7(2):197--215,
  2008.

\bibitem{ra06}
H.~{R}auhut.
\newblock {S}tability results for random sampling of sparse trigonometric
  polynomials.
\newblock {\em {I}{E}{E}{E} {T}rans. {I}nformation {T}heory},
  54(12):5661--5670, 2008.

\bibitem{rascva06}
H.~{R}auhut, K.~{S}chnass, and P.~{V}andergheynst.
\newblock {C}ompressed sensing and redundant dictionaries.
\newblock {\em {I}{E}{E}{E} {T}rans. {I}nform. {T}heory}, 54(5):2210 -- 2219,
  2008.

\bibitem{ruve06}
M.~{R}udelson and R.~{V}ershynin.
\newblock {O}n sparse reconstruction from {F}ourier and {G}aussian
  measurements.
\newblock {\em {C}omm. {P}ure {A}ppl. {M}ath.}, 61:1025--1045, 2008.

\bibitem{KV07}
K.~Schnass and P.~Vandergheynst.
\newblock Average performance analysis for thresholding.
\newblock {\em IEEE Signal Processing Letters}, 14(11):828--831, Nov. 2007.

\bibitem{hest03}
T.~{S}trohmer and R.~W. {H}eath.
\newblock {G}rassmannian frames with applications to coding and communication.
\newblock {\em {A}ppl. {C}omput. {H}armon. {A}nal.}, 14(3):257--275, 2003.

\bibitem{sz91}
S.~J. {S}zarek.
\newblock Condition numbers of random matrices.
\newblock {\em {J}. {C}omplexity}, 7:131--149, 1991.

\bibitem{tr04}
J.~A. {T}ropp.
\newblock {G}reed is good: {A}lgorithmic results for sparse approximation.
\newblock {\em {I}{E}{E}{E} {T}rans. {I}nform. {T}heory}, 50(10):2231--2242,
  2004.

\bibitem{tr05-1}
J.~A. {T}ropp.
\newblock {R}ecovery of short, complex linear combinations via $l_1$
  minimization.
\newblock {\em {I}{E}{E}{E} {T}rans. {I}nform. {T}heory}, 51(4):1568--1570,
  2005.

\bibitem{gisttr06-1}
J.~A. {T}ropp.
\newblock {A}lgorithms for simultaneous sparse approximation. {P}art {I}{I}:
  {C}onvex relaxation.
\newblock {\em {S}ignal {P}rocessing}, 86(3):589 -- 602, 2006.

\bibitem{tr06-2}
J.~A. {T}ropp.
\newblock {O}n the conditioning of random subdictionaries.
\newblock {\em {A}ppl. {C}omput. {H}armon. {A}nal.}, to appear.

\bibitem{gisttr06}
J.~A. {T}ropp, A.~C. {G}ilbert, and M.~J. {S}trauss.
\newblock {A}lgorithms for simultaneous sparse approximation. {P}art {I}:
  {G}reedy pursuit.
\newblock {\em {S}ignal {P}rocessing}, 86(3):572 -- 588, 2006.

\end{thebibliography}

\end{document}